\documentclass[11pt]{article}
\usepackage[margin=1in]{geometry}
\usepackage[T1]{fontenc}
\usepackage{booktabs} % For formal tables
\usepackage[ruled]{algorithm2e} % For algorithms

\SetAlFnt{\small}
\SetAlCapFnt{\small}
\SetAlCapNameFnt{\small}
\SetAlCapHSkip{0pt}
\IncMargin{-\parindent}
\usepackage{multirow}
\usepackage{hyperref}
\usepackage[utf8]{inputenc}
\usepackage{graphbox}
\usepackage{tikz}
\usepackage{lipsum}
\newcommand\BibTeX{{\rmfamily B\kern-.05em \textsc{i\kern-.025em b}\kern-.08em T\kern-.1667em\lower.7ex\hbox{E}\kern-.125emX}}
\usepackage{amsthm}

\usepackage[round]{natbib}

\usepackage{amsmath,amsfonts,bm, amssymb}

\def\eqref#1{equation~\ref{#1}}
\def\1{\bm{1}}

\DeclareMathAlphabet{\mathsfit}{\encodingdefault}{\sfdefault}{m}{sl}
\SetMathAlphabet{\mathsfit}{bold}{\encodingdefault}{\sfdefault}{bx}{n}

\def\gA{{\mathcal{A}}}

\def\gF{{\mathcal{F}}}

\def\gO{{\mathcal{O}}}

\def\sR{{\mathbb{R}}}

\def\sZ{{\mathbb{Z}}}

\DeclareMathOperator*{\argmax}{arg\,max}

\def\x{{\bf x}}

\def\0{{\bf 0}}
\def\1{{\bf 1}}

\def\argmax{\mathop{\rm argmax}}

\def\st{\mathsf{s.t.}}

\newcommand{\abs}[1]{|{#1}|}

\newcommand{\expect}{\mathbb{E}}

\newtheorem{thm}{Theorem}

\newtheorem{lem}{Lemma}

\begin{document}

\title{Principal-Agent Problem with Third Party: Information Design from Social Planner Perspective}
\author{Shiyun Lin\thanks{Center for Statistical Science, School of Mathematical Sciences, Peking University}, Zhihua Zhang\thanks{Center for Statistical Science, School of Mathematical Sciences, Peking University}}

\date{November, 2024}

\maketitle

\begin{abstract}
	We study the principal-agent problem with a third party that we call social planner, whose responsibility is to reconcile the conflicts of interest between the two players and induce socially optimal outcome in terms of some given social utility function. 
	The social planner owns no contractual power but manages to control the information flow between the principal and the agent.
	We design a simple workflow with two stages for the social planner. In the first stage, the problem is reformulated as an optimization problem whose solution is the optimal utility profile. In the second stage, we investigate information design and show that binary-signal information structure suffices to induce the socially optimal outcome determined in the first stage. 
	The simplicity and modularity of our method make it easy to implement in various scenarios within the principal-agent problem.
\end{abstract}

\section{Introduction}\label{sec:intro}
The principal-agent problem lies at the heart of modern economic theory due to its widespread applications, including corporate governance~\citep{young2008corporate,khan2011literature}, insurance design~\citep{pauly1968economics,vera2003structural}, healthcare systems~\citep{smith1997principal,scott1999patients}, contractual hiring arrangements~\citep{roach2016application}, education~\citep{levavcic2009teacher,lane2008interests}, real estate markets~\citep{anglin1991residential,pagliari2015principal}, sociology~\citep{adams1996principals}, etc. 
In classic principal-agent problems, there are two strategic entities involved in the system, namely the principal and the agent. 
Generally speaking, these two parties usually have different interests, such that the principal cannot directly ensure that the agent is always acting in the principal's best interest, and hence the principal would make a contract to alleviate the problem, which specifies the monetary transfer that the principal will pass to the agent as a function of the outcome.
However, due to information asymmetry (the agent having more information), contractual arrangement would not be perfect and there may exist welfare loss.

Moreover, in some cases there would be a third party acting as a mediator to reconcile the conflicts of interest between the principal and the agent. For example, in a publicly traded corporation, the relationship between shareholders and the management team can be modeled as a principal-agent problem, where shareholders (the principals) invest their money in the company and expect the management team (the agents) to make decisions that maximize the company's value.
However, conflicts of interest can arise when the management team has the power to make strategic decisions and compensation packages that might not align with shareholders' interests. 
Specifically, shareholders want the company to maximize profits, while the management team might be more focused on their own job security, reputation, or short-term financial gains.
To address this issue, a board of directors can act as a mediator (the third party) between the shareholders and the management team. The board is typically elected by shareholders and is responsible for overseeing management's actions, settling the compensation for the management and accounting to the shareholders for the organization's performance.

The above example can be abstracted as a principal-agent-mediator model, where the mediator acts as a third-party to mitigate the conflicts of interest between the principal and the agent. In this paper, we consider such mediator as a social planner, which is an independent third party. The social planner considers the principal and the agent as a whole, aiming to optimize the profit of the system and seek the socially optimal outcome so that both the principal and the agent would be satisfied. 
To achieve this goal, we leverage tools from information design, which is a technique to influence the outcome of a game by specifying the allocation of information~\cite{kamenica2019bayesian}.
In particular, during the interaction between the principal and the agent, the social planner possesses the power to manage the information flow between the two entities. In other words, by designing an appropriate information structure, the social planner controls how much information about the agent's action is revealed to the principal, so that the equilibrium of the Stackelberg game is socially optimal as defined by some chosen social utility function, given that both the principal and the agent are rational strategic players.

\citet{babichenko2022information} considered a similar problem and formalized the idea mathematically. In their paper, they characterized the implementability of the utility profiles of the principal and the agent; that is, they figured out the set of utility profiles that can be induced by some information structure. Based on their work, we take a step further and consider the optimization problem for the social planner to induce the socially optimal outcome for the system. 
We design a workflow for the social planner, which divides the task into two phases, by first solving an optimization problem to determine the socially optimal utility profile and then designing an information structure to guarantee the equilibrium of the game exactly lies in the chosen utility profile.
The two-stage formulation provides modularity and simplicity for the social planner. On one hand, the derivation of the socially optimal utility profile in the first stage could be solved by a simple geometric approach in the space of utility. On the other hand, the particular structure of the principal-agent problem greatly reduces the complexity of the strategic communication, in the sense that it suffices for the social planner to use binary signal to induce the utility profile determined in the first stage.

\subsection{Further Related Work}
The principal-agent problem has long been studied in the literature since the seminal work of \citet{ross1973economic,jensen1976theory,mirrlees1971exploration,holmstrom1979moral}. 
The effect of information revelation on the utility function of the principal and the agent has been discussed.
\citet{gjesdal1982information} studied a generalized agency model and provided a characterization for the ranking of information systems based on a generalization of Blackwell's ordering~\cite{blackwell1950comparison}, where the ranking is in terms of the principal's preference. \citet{kim1995efficiency} followed the agency framework and established another criterion based on mean-preserving spread. \citet{demougin2001ranking} later provided an integral condition and showed its equivalence to the mean-preserving spread condition. 
\citet{milgrom1981good} introduced the notion of ``favorableness'' of the revealed information in the sense of first-order stochastic dominance, whose focus is not only on the principal's, but also on the agent's point of view. 
\citet{jacobides2001information,silvers2012value,chaigneau2018does} further studied the effect of informativeness of signals on the favorableness, both for the principal and the agent.
Our concern is different from the above works, we focus on neither the principal's nor the agent's perspective, instead, we consider the two entities as a system and act as a social planner whose responsibility is to design an information structure that could induce a \emph{socially optimal} outcome.

Most of the literature studied the interaction between the two parties, namely the principal and the agent, nonetheless, there has been researches on the agency model involving a third-party.
\citet{maskin1990principal} introduced a third party whose optimization problem is to maximize an arbitrary weighted sum of the utility functions of different types of the principal. Moreover, this third party possesses the power of information transmission and contract implementation. 
\citet{braun1993governs} considered a model in political research with triadic structure, namely the policy maker, the intermediary funding organization, and the scientists, and they assessed the role of the third party in a qualitative way. 
\citet{van1998science, van2003new} further considered the tripartite relationship in an empirical study.
Recently, \citet{bizzotto2021communication} developed a model of information design with agency where the third party decides on the way of implementing the information structure proposed by the principal, \citet{boleslavsky2018bayesian} studied a variant of Bayesian persuasion where the underlying state is generated by a third party's unobservable effort, \citet{perez2022test} considered a setting where the state of the world could be falsified by a third party at a cost. 
Different from these works, the third party in our model has neither the contractual power, nor the power of determining the state of the world. The social planner in our model is responsible for formulating the information structure, which controls the information flow between the agent's action and the principal's belief.

\citet{mahzoon2023indicator} studied a variant of the principal-agent problem where the signal observed by the principal is chosen by the agent. They also considered the optimization problem with the goal of characterizing the optimal information structure from the agent's perspective, which could be viewed as a special case of our model if the social planner is fully in aligned with the agent's interest. They developed a geometric approach for solving the optimization problem. They constructed the geometric game in the space of likelihood ratios while our geometric approach works in the space of principal's and agent's utilities, which is more intuitive and direct. Interestingly, they also found that binary signal suffices for the optimal information structure.

\section{Problem Formulation}\label{sec:setting}
The \emph{principal-agent problem with third party} involves three strategic entities playing different roles in the game, 
where the \emph{agent} is the player who actually takes the action, 
the \emph{principal} possesses the contractual power, 
and the \emph{social planner} controls the information flow and tries to design the information about the agent's action that the principal will observe.
The common knowledge of the game consists of three components and is defined by a tuple $(n, r, c)$. The agent has $n$ possible actions, denoted as $[n] := \left\{0, 1, \cdots, n\right\}$. When the agent takes action $a$, a deterministic cost $c_a$ and a stochastic reward would be incurred.\footnote{Here we adopt the typical setting in the principal-agent problem, where the reward from an action is stochastic owing to some external factors.} Taking expectation over the possible outcomes, we denote the expected reward induced by action $a$ as $r_a$. Then the vector of rewards and costs are $r = (r_0, r_1, \cdots, r_n) \in \sR_{\geq 0}^{n+1}$ and $c = (c_0, c_1, \cdots, c_n) \in \sR_{\geq 0}^{n + 1}$, respectively. Note that the action $0$ is a special \emph{default} action: $r_0 = c_0 = 0$. Namely, the agent always has the option to take no effort and induce no reward.

Based on the knowledge of the rewards $r$ and the costs $c$, the social planner needs to design an information structure $(k, I)$, which is a stochastic mapping between the $n$ possible actions and the $k$ signals chosen by the social planner. Formally, $k \in \sZ_{+}$ is the number of possible \emph{signals} the principal may observe. And $I \in \sR^{n \times k}$ is a row-stochastic matrix, where the $a$-th row $I_a$ specifies the distribution over the $k$ signals when the agent takes action $a$.\footnote{Our model is similar as the one considered in \citet{babichenko2022information} in terms of the common knowledge of the game, a substantial difference is that in their setting, the information structure is given, and they aim to characterize whether a given information structure $(k, I)$ is implementable or not, while in our setting, we mainly aim to design an optimal information structure.}

There always exists an action which is the most ``cost-effective'', i.e., it has the maximal expected income for the principal among all the least costly actions for the agent (besides the default action 0). For notational simplicity  we denote this action as  $\hat{a}$. That is, $c_{\hat{a}} \leq c_a$ for $1 \leq a \leq n$, and $r_{\hat{a}} \geq r_a$ for every $a$ such that $c_a = c_{\hat{a}}$.

Upon observing the signal $j \in \left[k\right]$, the principal commits to transfer $t_j$ to the agent. We denote the contract by a vector $t = \left(t^1, t^2, \cdots, t^k\right)$ and let $t_0 := 0$ for convenience of notation.
Furthermore, we assume \emph{limited liability}~\citep{sappington1983limited,innes1990limited}: $t^j \geq 0$ for every $j \in \left[k\right]$, i.e., the principal cannot charge the agent. 
Therefore, the induced expected transfer at action $a \in \left[n\right]$ is denoted by $t_a := \expect_{j \sim I_a} \left[t^j\right]$. 

Through most of the paper, we consider a risk-neutral principal, whose utility function is given by $u_a^P = r_a - t_a$, i.e., the expected reward induced from the agent's action minus the expected monetary transfer. In Section~\ref{sec:risk_averse_principal}, we show similar result holds when the principal is risk-averse.
On the agent's side, we consider two types of risk attitudes.
On one hand, when the agent is \emph{risk-neutral}, her expected utility is linear, i.e., for an action $a$ and a contract $t$, the agent's utility is $u_a^A = t_a -  c_a$. 
On the other hand, when the agent is \emph{risk-averse}, we consider the utility function at action $a \in \left[n\right]$ being $u_a^A := \expect_{j \sim I_a} \left[v(t^j)\right] - c_a$, where $v: \sR_{\geq 0} \rightarrow \sR_{\geq 0}$ is a concave von Neumann-Morgenstern function that is used to capture the agent's attitude towards the part of utility deriving from monetary transfer.

The game with common knowledge $(n, r, c)$ proceeds in the following way:
\begin{enumerate}
	\item The social planner designs the information structure $(k, I)$ and informs the principal. 
	
	\item Based on the common knowledge $(n, r, c)$ and the information structure $(k, I)$, the principal designs the contract $t$ and informs the agent.
	
	\item The agent performs an action that can maximize her own utility function.
	
	\item The principal observes a signal related to the action performed by the agent based on the information structure $(k, I)$. According to the committed contract, a monetary transfer would be made between the principal and the agent. 
	
	\item The principal receives a reward and the agent suffers a cost, the utility function of the two parties are calculated accordingly.
\end{enumerate}

From the process above, we see that the principal-agent interaction is a Stackelberg game~\citep{stackelberg1952theory}. As is standard, we consider subgame perfect equilibria~\citep{selten1988reexamination}, and we naturally assume that both the principal and the agent are rational such that they will maximize their own utilities, and that the agent would break ties in the principal's favor. Under the above assumption, the utility profile at an equilibrium is uniquely determined.

In the above game, the goals of three players are as follows:
\begin{itemize}
	\item The agent: maximize her utility function based on the contract value $t$ and the cost $c$.
	
	\item The principal: design a contract that can incentivize the agent to take the action that would maximize the principal's utility.
	
	\item The social planner: design an information structure under which the Stackelberg game equilibrium is the one satisfying the social purpose, e.g., maximize a social utility function determined by the social planner.
\end{itemize}

In this paper, we would like to explore the following research question:

\emph{In terms of the social planner, what is the optimal information structure that can maximize a social utility function when the optimal contract is carried out by the principal?}

To answer this question, we design a workflow of the social planner, which consists of two phases.
Firstly, it utilizes common knowledge of the game to figure out an implementable utility profile that meets a specific social purpose, e.g., social welfare maximization or the utilities gained by the principal and the agent satisfy some fairness-aware criterion. Secondly, it designs an appropriate information structure to guide the behaviour of the principal and the agent towards the chosen utility profile. In Sections~\ref{sec:utility_profile_determination} and \ref{sec:info_struct_design}, we discuss these two phases, respectively, covering the cases where the agent is risk-neutral or risk-averse.

\section{Determination of Utility Profile}\label{sec:utility_profile_determination}
The determination of the utility profile can be reduced to an optimization problem that needs to be solved by the social planner, where the optimization objective and the feasible region are defined by the social utility function and the sets of implementable utility profiles, respectively.

\subsection{Social Utility Function}\label{subsec:social_welfare}
The social planner works with a system with two individuals, namely the principal and the agent.
Therefore, the purpose of the social planner is to maximize a function $w\colon \sR^2 \rightarrow \sR$ that aggregates each profile $(x, y) \in \sR^2$ of individual utility values into a social utility, where $u_a^P = x$, $u_a^A = y$. 
Technically, for our framework to be well-defined, we only require $w$ to be jointly continuous in $(u_a^P, u_a^A)$.
Practically, in this paper, to properly define the social utility function, we consider concepts from multi-agent resource allocation (MARA)~\citep{chevaleyre2006issues}.

\subsubsection{Overall Utility Based Function}\label{subsubsec:overall_utlity}
The overall utility based social utility function aims to measure the quality of the utility profile from the viewpoint of the system as a whole, and hence would consider a notion where each agent would have a contribution to the social utility function. Here we consider two widely adopted notions in \emph{Welfare Economics} and \emph{Social Choice Theory}.

\paragraph{Utilitarian social welfare} 
The concept of \emph{utilitarian social welfare} is defined as the sum of individual utilities:
\begin{align}\label{eq:utilitarian_sw}
	w_{USF}(x, y) = x + y.
\end{align}
This notion is probably the most widely used interpretation of the term ``social welfare'' in the multiagent systems literature~\citep{wooldridge2009introduction,sandholm1999distributed}, and it can provide a suitable metric for overall (as well as average) profit.

\paragraph{Nash product}
The \emph{Nash product} is defined as the product of individual utilities:
\begin{align}\label{eq:nash_prod}
	w_{NP}(x, y) = x \cdot y.
\end{align}
The notion of Nash product favours both increases in overall utility and inequality-reducing redistributions. Therefore, it would be a compromise between the utilitarian and egalitarian social welfare (which is described in the following).

\subsubsection{Equality Based Function}\label{subsubsec:equality}
The equality based social utility function is dedicated to address fairness between the principal and the agent.

\paragraph{Egalitarian Social Welfare}
The \emph{egalitarian social welfare} is given by the utility of the agent that is currently worse off. Therefore, in the principal-agent problem, we have
\begin{align}\label{eq:egalitarian_sw}
	w_{ESF}(x, y) = \min \left\{x, y\right\}.
\end{align}
The above notion is usually considered in the area of fair division~\citep{young1995equity,brams1996fair,moulin2004fair}, which offers a level of \emph{fairness} and indicates that the system should satisfy a minimum need of the two agents.

\paragraph{Approximated Fairness}
The \emph{approximated fairness}~\citep{fujita2012secure} ranks utility profiles based on the squared sum of the deviation of individual utility from average of the utilities:
\begin{align}\label{eq:approximate}
	%\tilde{w}_{AF}(x, y) = \sum_{i = 1}^{n} \frac{(u_i - \bar{u})^2}{n} 
	\tilde{w}_{AF}(x, y) = \frac{1}{2}  \left(x - \frac{x + y}{2}\right)^2 + \frac{1}{2}  \left(y - \frac{x + y}{2}\right)^2 = \frac{\left(x - y\right)^2}{4}.
\end{align}
A utility profile $(x, y)$ is considered to be ideal if its $\tilde{w}_{AF}(x, y) = 0$, which is only achieved when $x = y$. 
Therefore, as a social utility function that the social planner would like to maximize it, we define $w_{AF}(x, y) = - \tilde{w}_{AF}(x, y)$.

All four utility functions listed above are jointly continuous in $x$ and $y$, i.e., the utility of the principal and the agent.

\subsection{Implementable Utility Profiles}\label{subsec:implement_utility}
The implementability of a utility profile $(x, y)$ is essential for information design in the principal-agent problem, as it characterizes the feasible solution of the optimization problem for the social planner. 
We say that a utility profile $(x, y) \in \sR_{\geq 0} \times \sR_{\geq 0}$ is implementable if there exists an information structure $I$ such that given $I$, the equilibrium outcome of the Stackelberg game is exactly $(x, y)$. \citet{babichenko2022information} characterize the set of implementable utility profiles, which is simple for the case that the risk attitude of the agent is either risk-neutral or risk-averse. The implementable utility profile set can be described by thresholds on the utilities of the two individuals in the system.

\paragraph{Risk-neutral Agent}~\citep[Theorem 3.1]{babichenko2022information} 
We first define the set of possible utility profiles of a given action $a$ for a risk-neutral agent:
\begin{align}\label{possible_set_risk_neutral}
	\begin{split}
		F_a := \left\{(x, y): x = r_a - s, y = -c_a + s, s \geq 0\right\}.
	\end{split}
\end{align}
Denote by $F := \cup_{1 \leq a \leq n} F_a \cup \left\{(0, 0)\right\}$ the super set of all possible utility profiles.
The implementable set of utility profiles with a risk-neutral agent can be described as 
\begin{align}\label{eq:implentable_set_risk_neutral}
	\gF := \left\{(x, y) \in F:  x \geq \max \left\{0, r_{\hat{a}} - c_{\hat{a}}\right\}, y \geq 0\right\}.
\end{align}

\paragraph{Risk-averse Agent}~\citep[Theorem 4.1]{babichenko2022information}
We first define the set of possible utility profiles of a given action $a$ for a risk-averse agent:
\begin{align}\label{eq:possible_set_risk_averse} 
	\begin{split}
		F_a := &\left\{(x, y): x = r_a - z, y \leq v(z) - c_a \quad \text{for some } z \in \sR_{\geq 0}\right\} \\
		= &\left\{(x, y): x \leq r_a, y \leq v(r_a - x) - c_i\right\}.
	\end{split}
\end{align}
Denote by $F := \cup_{1 \leq a \leq n} F_a \cup \left\{(0, 0)\right\}$ the super set of all possible utility profiles.
The implementable set of utility profiles with a risk-averse agent can be described as 
\begin{align}\label{eq:implementable_set_risk_averse}
	\gF := \left\{(x, y) \in F: x \geq \max \left\{r_{\hat{a}} - v^{-1}(c_{\hat{a}}), 0\right\}, y \geq 0\right\}.
\end{align}

The implementable sets for risk-neutral and risk-averse agent are shown in Figure~\ref{fig:implementable_set}.
From the characterization of the implementable sets (Eq.~(\ref{eq:implentable_set_risk_neutral}), (\ref{eq:implementable_set_risk_averse}) and Figure~\ref{fig:implementable_set}), we see that the set of implementable utility profiles is bounded and closed, i.e., the feasible region of the optimization problem that we are interested in (as defined in Eq.(\ref{eq:optimization}) below) is bounded and closed.

\vspace{-0.2in}
\begin{figure}[!htp]
	\centering
	\begin{tikzpicture}[xscale=0.6, yscale=0.6]
		\draw (-1,0)--(2,0);
		\draw[red!60!gray] (2,0)--(7,0);
		\draw[->] (7,0)--(8,0);
		\node[right] at (8,0) {$u^P$};
		
		\draw[->] (0,-3.5)--(0,6);
		\node[above] at (0,6) {$u^A$};
		
		\draw (4,-2)--(-1,3);
		\draw (5,-4)--(-1,2);
		\draw (6,-2.5)--(-1,4.5);
		\draw (8,-3)--(-1,6);
		
		\draw[line width=1mm] (3.5,0)--(2,1.5);
		
		\draw[line width=1mm] (5,0)--(2,3);
		
		\filldraw (4,-2) circle(0.1);
		\node[below] at (4.2,-2) {\footnotesize $(r_{\hat{a}}, -c_{\hat{a}})$}; 
		
		\filldraw (8,-3) circle(0.1);
		\node[below] at (8,-3) {\footnotesize $(r_a, -c_a)$};
		
		\filldraw (2,0) circle(0.1);
		\node[below] at (2,-0.1) {\footnotesize $r_{\hat{a}} - c_{\hat{a}}$};
		
		\filldraw[color = red!50!gray, fill = red!50!gray] (5, 0) circle(0.1);
		\node[below] at (5,-0.1) {\footnotesize $r_a - c_a$};
		
		\filldraw[color = red!50!gray, fill = red!50!gray] (2, 3) circle(0.1);
		\node[right] at (2,3) {\footnotesize $(r_{\hat{a}} - c_{\hat{a}}, r_a - c_a - r_{\hat{a}} + c_{\hat{a}})$};

		\draw[red!60!gray] (2,0)--(2,6);
		
		\node at (4.2, -4.5) {(a)};
	\end{tikzpicture}
	\hspace{7mm}
	\begin{tikzpicture}[xscale=0.7, yscale=1.4]
		\fill [red!40, domain = 1.001:2.64, variable = \x]
		(2.36, 0)--
		plot ({-\x+5}, {sqrt{\x} - 1})--cycle;
		
		\fill [red!40, domain = 2.56:3.868, variable = \x]
		(3.132, 0)--
		plot ({-\x+7}, {sqrt{\x} - 1.6})--cycle;
		
		\draw[->] (-1, 0)--(8,0);
		\draw[->] (0, -1.5)--(0,2);
		\node[right] at (8,0) {$u^P$};
		\node[above] at (0,2) {$u^A$};
		
		\draw[dashed, domain = 0.01:4, variable = \x]
		plot ({-\x+3}, {sqrt{\x} - 0.8});
		
		\draw[dashed, domain = 0.01:6, variable = \x]
		plot ({-\x+5}, {sqrt{\x} - 1});
		
		\draw[dashed, domain = 0.01:8, variable = \x]
		plot ({-\x+7}, {sqrt{\x} - 1.6});
		
		\filldraw (2.99, -0.7) ellipse (0.08 and 0.04);
		\node[below] at (2.99, -0.75) {\footnotesize $(r_{\hat{a}}, -c_{\hat{a}})$};
		
		\filldraw (4.99, -0.9) ellipse (0.08 and 0.04);
		\node[below] at (4.99, -0.95) {\footnotesize $(r_a, -c_a)$};
		
		\filldraw (6.99, -1.5) ellipse (0.08 and 0.04);
		
		\filldraw (2.36, 0) ellipse (0.08 and 0.04);
		\node[below] at (2.36, -0.05) {\footnotesize $r_{\hat{a}} - v^{-1}(c_{\hat{a}})$};
		
		\draw[domain = 0.01:3.868, variable = \x]
		plot ({-\x+7}, {sqrt{\x} - 1.6});
		
		\draw[domain = 1.868:6, variable = \x]
		plot ({-\x+5}, {sqrt{\x} - 1});
		
		\node at (4.2, -2.25) {(b)};
	\end{tikzpicture}
	\caption{Utility profiles and the implementable set. (a) The agent is risk-neutral, the bold parts of the lines are the implementable utility profiles. (b) The agent is risk-averse, the area below the solid line is the super set of all possible utility profiles, while the pink area is the set of implementable utility profiles.}
	\label{fig:implementable_set}
\end{figure}
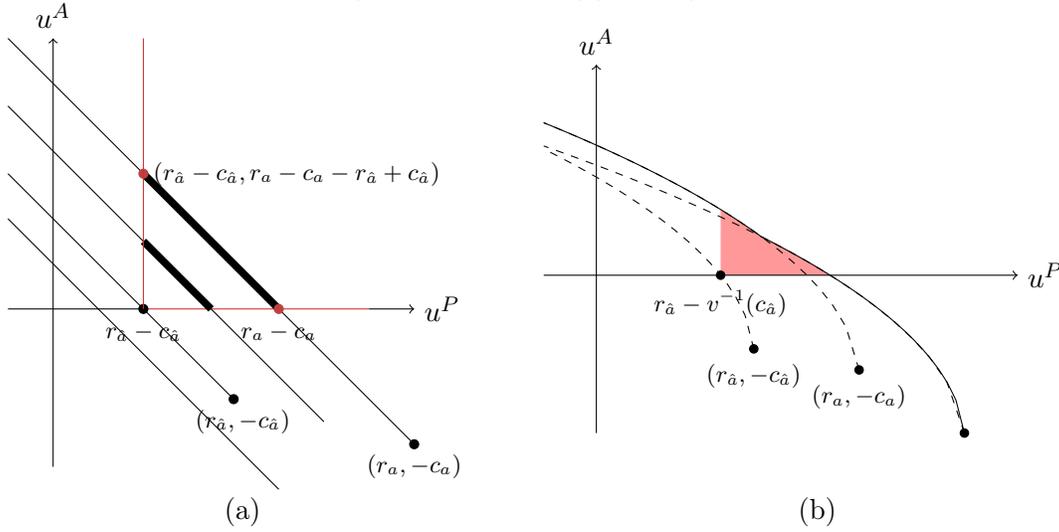
\vspace{-0.2in}

Moreover, we say that an action $a$ is implementable if there exists some information structure $I$ such that under this information structure, the agent would choose action $a$ assuming that the principal is rational in the sense that she would choose the optimal contract, i.e., action $a$ is at the equilibrium of the Stackelberg game.
For a risk-neutral agent, an action $a \in [n]$ is implementable if and only if $r_a - c_a \geq \max \left\{r_{\hat{a}} - c_{\hat{a}}, 0\right\}$, while for a risk-averse agent, action $a \in [n]$ is implementable if and only if $r_a - v^{-1}(c_a) \geq \max \left\{r_{\hat{a}} - v^{-1}(c_{\hat{a}}), 0\right\}$~\citep{babichenko2022information}.

\subsection{Maximize the Social Utility Function over Implementable Set}\label{subsec:optimization}
Based on the social utility function and the implementable utility profiles, the social planner needs to solve the following constrained optimization problem 
\begin{align}\label{eq:optimization}
	\max_{(x, y) \in \gF} w(x, y),
\end{align}
whose solution should be the input of the next phase, i.e., the equilibrium of the Stackelberg game that the social planner would like to induce.

In general, to ensure that the constrained optimization problem is solvable, by Weierstrass extreme value theorem, we require the feasible region $\gF$ to be bounded and closed, and $w$ to be jointly continuous in both $x$ and $y$. From Section~\ref{subsec:social_welfare} and \ref{subsec:implement_utility}, we see that the social utility functions and the implementable utility profiles under consideration satisfy the technical requirements and hence the optimal solution is guaranteed to exist.
However, in practice, solving the optimization problem~(\ref{eq:optimization}) is not simple. 
On one hand, the feasible region, i.e., the set of the implementable utility profiles, is not convex, both in the regime with a risk-neutral agent and a risk-averse agent (Figure~\ref{fig:implementable_set}).
On the other hand, the social utility function itself may not be concave in some case, e.g., the Nash product.
Therefore, we discuss case by case for different social utility functions and different risk attitudes of the agent.

\subsubsection{Overall Utility-Based Social Utility Function}\label{subsubsec:overall_optimization}
\paragraph{Risk-neutral Agent}
If the agent is risk-neutral, then no matter the social utility function is utilitarian social welfare (USF) or Nash product (NP), the social planner can use USF as a criterion to determine the induced action. 

To see the correctness of this claim, we first show that for each implementable action, there is a corresponding unique utility profile that maximizes the Nash product.
When the action $a$ is given, the relationship between the utility of the principal and the agent can be described as $u_a^A + u_a^P = r_a - c_a$. 
Then the Nash product of a utility profile is $u_a^A \cdot u_a^P = u_a^P \cdot (r_a - c_a - u_a^P)$, which is a quadratic function in terms of $u_a^P$. Note that $\max \left\{0, r_{\hat{a}} - c_{\hat{a}}\right\} \leq u_a^P \leq r_a - c_a$, by simple calculation, we know that the maximum of the Nash product attains at
\begin{align*}
	u_a^P = \left\{
	\begin{array}{ll}
		\frac{r_a - c_a}{2} & \text{if } \frac{r_a - c_a}{2} \geq \max \left\{0, r_{\hat{a}} - c_{\hat{a}}\right\}, \\
		r_{\hat{a}} - c_{\hat{a}} & \text{otherwise,}
	\end{array}
	\right.
\end{align*}
and 
\begin{align*}
	u_a^A = \left\{
	\begin{array}{ll}
		\frac{r_a - c_a}{2} & \text{if } \frac{r_a - c_a}{2} \geq \max \left\{0, r_{\hat{a}} - c_{\hat{a}}\right\}, \\
		r_a - c_a - \max \left\{0, r_{\hat{a}} - c_{\hat{a}}\right\} & \text{otherwise,}
	\end{array}
	\right.
\end{align*}
which proves the uniqueness of the utility profile that achieves the maximum Nash product, and its corresponding Nash product, denoted as $NP^*_{a}$ would be
\begin{align*}
	\left\{
	\begin{array}{ll}
		\frac{(r_a - c_a)^2}{4} & \text{if } \frac{r_a - c_a}{2} \geq \max \left\{0, r_{\hat{a}} - c_{\hat{a}}\right\}, \\
		\max \left\{0, r_{\hat{a}} - c_{\hat{a}}\right\} \cdot (r_a - c_a - \max \left\{0, r_{\hat{a}} - c_{\hat{a}}\right\}) & \text{otherwise.}
	\end{array}
	\right.
\end{align*}
For two different actions $a$ and $a'$, we show that $USF_a \leq USF_{a'}$ implies $NP^*_a \leq NP^*_{a'}$ in the following, which further implies that the social planner can use $USF$ to determine the action then search for the corresponding utility profile, simplifying the optimization process when dealing with the non-convex objective function Nash product.

If $USF_a \leq USF_{a'}$, i.e., $r_a - c_a \leq r_{a'} - c_{a'}$, there would be three possible cases: 
\begin{enumerate}
	\item $\max \left\{0, r_{\hat{a}} - c_{\hat{a}}\right\} \leq \frac{r_a - c_a}{2} \leq \frac{r_{a'} - c_{a'}}{2}$: The corresponding $NP^*$ would be $NP^*_a = \frac{(r_a - c_a)^2}{4} \leq \frac{(r_{a'} - c_{a'})^2}{4} = NP^*_{a'}$.
	
	\item $\frac{r_a - c_a}{2} \leq \frac{r_{a'} - c_{a'}}{2} \leq \max \left\{0, r_{\hat{a}} - c_{\hat{a}}\right\}$: The the corresponding $NP^*$ would be $NP^*_a = \max \left\{0, r_{\hat{a}} - c_{\hat{a}}\right\} \cdot (r_a - c_a - \max \left\{0, r_{\hat{a}} - c_{\hat{a}}\right\}) \leq \max \left\{0, r_{\hat{a}} - c_{\hat{a}}\right\} \cdot (r_{a'} - c_{a'} - \max \left\{0, r_{\hat{a}} - c_{\hat{a}}\right\}) = NP^*_{a'}$, since we have $\max \left\{0, r_{\hat{a}} - c_{\hat{a}}\right\} \geq 0$ and $r_{a'} - c_{a'} \geq r_a - c_a \geq \max \left\{0, r_{\hat{a}} - c_{\hat{a}}\right\}$ by the implementability of the actions.
	
	\item $\frac{r_a - c_a}{2} \leq \max \left\{0, r_{\hat{a}} - c_{\hat{a}}\right\} \leq \frac{r_{a'} - c_{a'}}{2}$: In this case, $NP^*_a = \max \left\{0, r_{\hat{a}} - c_{\hat{a}}\right\} \cdot (r_a - c_a - \max \left\{0, r_{\hat{a}} - c_{\hat{a}}\right\})$, while $NP^*_{a'} = \frac{(r_{a'} - c_{a'})^2}{4}$. Since $\max u_a^P \cdot \left(r_a - c_a - u_a^P\right) = \frac{(r_a - c_a)^2}{4}$, we have that $NP^*_a \leq \frac{(r_a - c_a)^2}{4} \leq NP^*_{a'}$.
\end{enumerate}

For utilitarian social welfare (USF), the optimal action for the agent is unique from the social planner's perspective, while the utility profile is not necessarily unique, since the monetary transfer between the principal and the agent would not affect the USF. Therefore, if the social planner chooses USF as the social utility function, any utility profile that satisfies $u^A + u^P = \max_{a \in \left[n\right]} \left\{r_a - c_a\right\}$ could be the optimal utility profile and the social planner can use it as an input for the next stage of information structure design.

For Nash product (NP), both the optimal action and the optimal utility profile is unique, as is shown in the above. The determined utility profile is
\begin{align}\label{eq:utility_profile_np_risk_neutral}
	(u^P, u^A) = \left\{
	\begin{array}{ll}
		\left(\frac{r_{a^*} - c_{a^*}}{2}, \frac{r_{a^*} - c_{a^*}}{2}\right) & \text{if } \frac{r_{a^*} - c_{a^*}}{2} \geq \max \left\{0, r_{\hat{a}} - c_{\hat{a}}\right\}, \\
		\left(r_{\hat{a}} - c_{\hat{a}}, r_{a^*} - c_{a^*} - \max \left\{0, r_{\hat{a}} - c_{\hat{a}}\right\}\right) & \text{otherwise,}
	\end{array}
	\right.
\end{align}
where $a^* = \argmax_{a \in \left[n\right]} \left\{r_{a} - c_{a}\right\}$. The optimal utility profile under the two cases are shown in Figure~\ref{fig:optimization_overall_risk_neutral}.

\begin{figure}[!htp]
	\centering
	\begin{tikzpicture}[xscale=0.6, yscale=0.6]
		\draw (-1,0)--(2,0);
		\draw[red!60!gray] (2,0)--(7,0);
		\draw[->] (7,0)--(8,0);
		\node[right] at (8,0) {$u^P$};
		
		\draw[->] (0,-3.5)--(0,6);
		\node[above] at (0,6) {$u^A$};
		
		\draw (4,-2)--(-1,3);
		\draw (5,-4)--(-1,2);
		\draw (6,-2.5)--(-1,4.5);
		\draw (8,-3)--(-1,6);
		
		\draw[line width=0.8mm] (3.5,0)--(2,1.5);
		
		\draw[line width=0.8mm] (5,0)--(2,3);
		
		\filldraw (4,-2) circle(0.1);
		\node[below] at (4.2,-2) {\footnotesize $(r_{\hat{a}}, -c_{\hat{a}})$}; 
		
		\filldraw (8,-3) circle(0.1);
		\node[below] at (8,-3) {\footnotesize $(r_{a^*}, -c_{a^*})$};
		
		\filldraw (2,0) circle(0.1);
		\node[below] at (2,-0.1) {\footnotesize $r_{\hat{a}} - c_{\hat{a}}$};
		
		\filldraw[color = red!50!gray, fill = red!50!gray] (5, 0) circle(0.1);
		\node[below] at (5,-0.1) {\footnotesize $r_{a^*} - c_{a^*}$};
		
		\filldraw[color = red!50!gray, fill = red!50!gray] (2, 3) circle(0.1);

		\draw[red!60!gray] (2,0)--(2,6);
		
		\draw [domain = 1.04:7.5, variable = \x]
		plot ({\x}, {6.25 / \x});
		
		\filldraw[color = blue!50, fill = blue!50] (2.5, 2.5) circle(0.1);
		\node[right] at (2.55, 2.5) {\footnotesize $(\frac{r_{a^*} - c_{a^*}}{2}, \frac{r_{a^*} - c_{a^*}}{2})$};
		
		\node at (4.2, -4.5) {(a)};
	\end{tikzpicture}
	\hspace{5mm}
	\begin{tikzpicture}[xscale=0.6, yscale=0.6]
		\draw (-1,0)--(3.5,0);
		\draw[red!60!gray] (3.5,0)--(7,0);
		\draw[->] (7,0)--(8,0);
		\node[right] at (8,0) {$u^P$};
		
		\draw[->] (0,-3.5)--(0,6);
		\node[above] at (0,6) {$u^A$};
		
		\draw (5.5,-2)--(-1,4.5);
		\draw (5,-4)--(-1,2);
		\draw (6.5,-2.5)--(-0.5,4.5);
		\draw (8,-3)--(-1,6);
		
		\draw[line width=0.8mm] (4,0)--(3.5,0.5);
		
		\draw[line width=0.8mm] (5,0)--(3.5,1.5);
		
		\filldraw (6.5,-2.5) circle(0.1);
		\node[below] at (6.7,-2.5) {\footnotesize $(r_{\hat{a}}, -c_{\hat{a}})$}; 
		
		\filldraw (8,-3) circle(0.1);
		\node[below] at (8,-3) {\footnotesize $(r_{a^*}, -c_{a^*})$};
		
		\filldraw (3.5,0) circle(0.1);
		\node[below] at (3,-0.1) {\footnotesize $r_{\hat{a}} - c_{\hat{a}}$};
		
		\filldraw[color = red!50!gray, fill = red!50!gray] (5, 0) circle(0.1);
		\node[below] at (5.5,-0.1) {\footnotesize $r_{a^*} - c_{a^*}$};
		
		\draw[red!60!gray] (3.5,0)--(3.5,6);
		
		\draw [dashed, color = blue, domain = 1.04:7.5, variable = \x]
		plot ({\x}, {6.25 / \x});
		
		\draw [domain = 0.88:7.5, variable = \x]
		plot ({\x}, {5.25 / \x});
		
		\filldraw[color = blue!50, fill = blue!50] (2.5, 2.5) circle(0.1);
		
		\filldraw[color = green!50!black, fill = green!50!black] (3.5, 1.5) circle(0.1);
		\node[right] at (3.5,1.7) {\footnotesize $(r_{\hat{a}} - c_{\hat{a}}, r_{a^*} - c_{a^*} - r_{\hat{a}} + c_{\hat{a}})$};
		
		\node at (4.2, -4.5) {(b)};
	\end{tikzpicture}
	\caption{Optimal utility profile if the social planner uses Nash product as the social utility function and the agent is risk-neutral. Denote $a^*$ as the action that induces the largest utilitarian social welfare (USF). (a) $\frac{r_{a^*} - c_{a^*}}{2} \geq \max \left\{0, r_{\hat{a}} - c_{\hat{a}}\right\}$, the blue point represents the optimal utility profile. (b) $0 \leq \frac{r_{a^*} - c_{a^*}}{2} \leq r_{\hat{a}} - c_{\hat{a}}$, the green point represents the optimal utility profile.}
	\label{fig:optimization_overall_risk_neutral}
\end{figure}
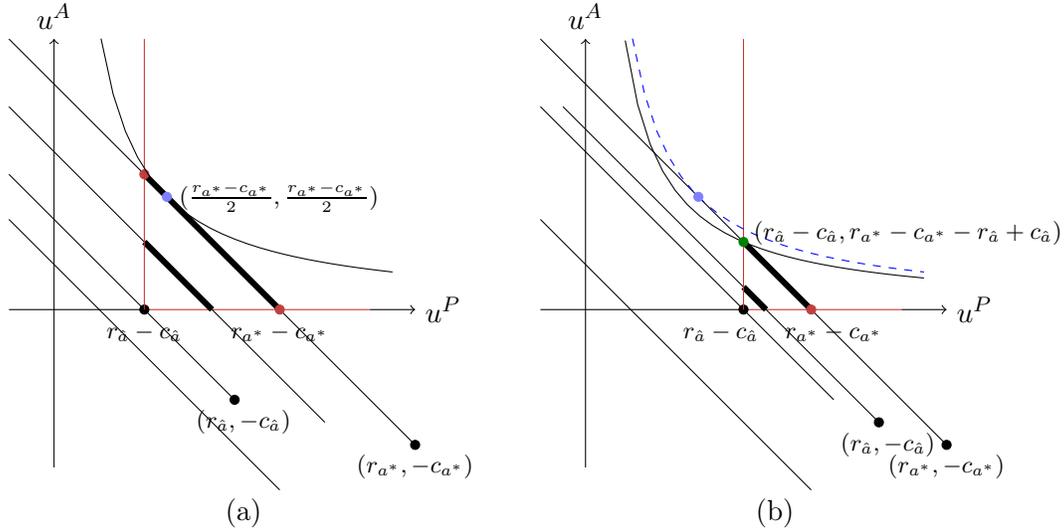

\paragraph{Risk-averse Agent}
If the agent is risk-averse, the set of implementable utility pairs becomes significantly richer compared with the risk-neutral case, and the whole set itself is not convex, the optimization problem may become complicated.
To address this issue, we can decompose problem~(\ref{eq:optimization}) into $n$ subproblems, where $n$ is the number of implementable actions, and the following optimization problem is equivalent to the original problem:
\begin{align}\label{eq:sub_optimization}
	\max_{a \in [n]} \max_{(x, y) \in F_{a} \cap \gF} w(x, y).	
\end{align}
Here, the outer problem is a maximization problem over $n$ candidates, which can be done in $\gO(n)$ comparisons, while the inner problem is an optimization problem over a convex set, since for every given action $a$, $u_a^A = v(r_a - u_a^P) - c_a$ is a concave function by the concavity of $v$, and the hypograph of the function $u_a^A$ is a convex set.

For utilitarian social welfare, the objective function is linear, its maximum over a convex set must exist on the boundary. Therefore, the solution can be reduced to finding the tangent point on the boundary with a slope of -1, i.e., the inner optimization problem is equivalent to solve for the maximum $c$ such that $y = -x + c$ intersects $F_a \cap \gF$. The solution is as follows, while the geometric illustration is shown in Figure~\ref{fig:optimization_usf_risk_averse} in Appendix~\ref{sec:opt_profile_usf_averse}.
\begin{align}\label{eq:utility_profile_usf_risk_averse}
	\begin{split}
		\left(u_a^P, u_a^A\right) = \left\{
		\begin{array}{ll}
			\left(r_a - v^{-1}(c_a), 0\right) & \text{if } v'(v^{-1}(c_a)) \leq 1, \\
			\multirow{2}{*}{$\left(r_{\hat{a}} - v^{-1}(c_{\hat{a}}), v(r_a - r_{\hat{a}} + v^{-1}(c_{\hat{a}})) - c_a\right)$}
			& \text{if } r_{\hat{a}} - v^{-1}(c_{\hat{a}}) \geq 0  \\
			& \text{ and } v'(r_a - r_{\hat{a}} + v^{-1}(c_{\hat{a}})) > 1,\\
			\left(0, v(r_a) - c_a\right) & \text{if } r_{\hat{a}} - v^{-1}(c_{\hat{a}}) < 0 \text{ and } v'(r_a) > 1,\\
			\left(r_a -v'^{(-1)}(1), v(v'^{(-1)}(1) - c_a)\right) & \text{otherwise.} \\
		\end{array}
		\right.
	\end{split}
\end{align}
To show the correctness of the solution, we only need to notice that $v$ is a strictly concave function, so that its first order derivative is strictly decreasing and the equation $v'(x) = 1$ has unique solution.

For Nash product, although the function $w(x, y) = x \cdot y$ is not concave, the optimal utility profile of a given action $a$ still has a simple form. By the individual rationality of the principal and the agent, we have that $u^A, u^P \geq 0$, hence when $u^P$ is fixed, the larger $u^A$, the larger the Nash product, which implies that the utility profile attains the maximum Nash product on the boundary $y = v(r_a - x) - c_a$ of the implementable set, as in the case of utilitarian social welfare. 
Let $x_a$ be the solution of the equation $v(r_a - x) - x \cdot v'(r_a - x) = c_a, y_a = v(r_a - x_a) - c_a$. 
And let $x_1 = r_{\hat{a}} - v^{-1}(c_{\hat{a}})$, 
the solution in this scenario is as follows
\begin{align}\label{eq:utility_profile_np_risk_averse}
	\begin{split}
		\left(u_a^P, u_a^A\right) = \left\{
		\begin{array}{ll}
			\left(x_a, y_a\right) & \text{if } x_a \in \left[\max \left\{0, r_{\hat{a}} - v^{-1}(c_{\hat{a}})\right\}, r_a - v^{-1}(c_a)\right], \\
			%			\left(0, \right) & \text{if } r_{\hat{a}} - v^{-1}(c_{\hat{a}}) \leq 0 \text{ and } x_a \leq 0, \\
			\left(x_1, v(r_a - x_1) - c_a\right) & \text{if } r_{\hat{a}} - v^{-1}(c_{\hat{a}}) \geq 0 \text{ and } 0 \leq x_a \leq r_{\hat{a}} - v^{-1}(c_{\hat{a}}). \\
			%			\left(x_a, y_a\right) & \text{if } r_{\hat{a}} - v^{-1}(c_{\hat{a}}) < 0.
		\end{array}
		\right.
	\end{split}
\end{align}
We prove that Eq. (\ref{eq:utility_profile_np_risk_averse}) provides the maximum Nash product over all implementable utility profiles for the given action $a$, i.e., $F_a \cap \gF$, in Appendix~\ref{sec:np_risk_averse}. For graph illustration, see Figure~\ref{fig:optimization_np_risk_averse} in Appendix~\ref{sec:opt_profile_np_averse}.

\subsubsection{Equality-based Social Utility Function}\label{subsubsec:equality_optimization}
\paragraph{Egalitarian Social Welfare}
When the agent is \emph{risk-neutral}, for a given action $a$, the relationship between the utility function of the principal $u_a^P$ and the agent $u_a^A$ is linear, which satisfies $u_a^A = -u_a^P + (r_a - c_a)$. 
Therefore, for a given action $a$, we have
\begin{align*}
	\max_{u_a^A + u_a^P = r_a - c_a} \min \left\{u_a^P, u_a^A\right\} = \frac{r_a - c_a}{2}.
\end{align*}
Hence, if $(\frac{r_a - c_a}{2}, \frac{r_a - c_a}{2})$ is implementable, i.e., $\frac{r_a - c_a}{2} \geq \max \left\{0, r_{\hat{a}} - c_{\hat{a}}\right\}$, it would be the optimal utility profile for the given action $a$ and the corresponding ESF is $\frac{r_a - c_a}{2}$. 
On the other hand, if the above utility profile is not implementable, by the definition of implementable set for a risk-neutral agent (Eq.(\ref{eq:implentable_set_risk_neutral})), we know that $r_{\hat{a}} - c_{\hat{a}} > 0$, $\frac{r_a - c_a}{2} < r_{\hat{a}} - c_{\hat{a}}$, and for any $(u_a^P, u_a^A) \in \gF$ such that the induced action is $a$, we have that $\min \left\{u_a^P, u_a^A\right\} = u_a^A$. 
Therefore, in this case, the maximum egalitarian social welfare the social planner can expect is $r_a - c_a - r_{\hat{a}} + c_{\hat{a}}$ which occurs at the utility profile $(r_{\hat{a}} - c_{\hat{a}}, r_a - c_a - r_{\hat{a}} + c_{\hat{a}})$.
Furthermore, consider two different implementable actions $a$ and $a'$, without loss of generality we assume that $r_a - c_a \leq r_{a'} - c_{a'}$ we have three possible scenarios for the optimal utility profiles with respect to the two actions:
\begin{itemize}
	\item If $\left(\frac{r_a - c_a}{2}, \frac{r_a - c_a}{2}\right)$ is implementable, then we must have $\left(\frac{r_{a'} - c_{a'}}{2}, \frac{r_{a'} - c_{a'}}{2}\right)$ is also implementable, since $\max \left\{0, r_{\hat{a}} - c_{\hat{a}}\right\} \leq \frac{r_a - c_a}{2} \leq \frac{r_{a'} - c_{a'}}{2}$. In this scenario, $a'$ would be preferred as it returns the larger ESF.
	
	\item If $\left(\frac{r_{a'} - c_{a'}}{2}, \frac{r_{a'} - c_{a'}}{2}\right)$ is not implementable, then we must have $\left(\frac{r_a - c_a}{2}, \frac{r_a - c_a}{2}\right)$ is neither implementable for the similar reason as above. In this scenario, the ESF for the two actions are $r_a - c_a - r_{\hat{a}} + c_{\hat{a}} \leq r_{a'} - c_{a'} - r_{\hat{a}} + c_{\hat{a}}$ and hence the preferred action is also $a'$.
	
	\item If $\left(\frac{r_{a'} - c_{a'}}{2}, \frac{r_{a'} - c_{a'}}{2}\right)$ is implementable, while $\left(\frac{r_a - c_a}{2}, \frac{r_a - c_a}{2}\right)$ is not implementable, the preferred action is still $a'$. In this case, we have $0 \leq \frac{r_a - c_a}{2} \leq r_{\hat{a}} - c_{\hat{a}} \leq \frac{r_{a'} - c_{a'}}{2}$, the egalitarian social welfare for action $a$ would be $r_a - c_a - r_{\hat{a}} + c_{\hat{a}} \leq r_a - c_a - \left(\frac{r_a - c_a}{2}\right) = \frac{r_a - c_a}{2} \leq \frac{r_{a'} - c_{a'}}{2}$.
\end{itemize}
Based on the above observations, we can conclude that with a risk-neutral agent, if the social planner uses ESF as the social utility function, the induced action for the agent is the one with maximum utilitarian social welfare, i.e., $a^*= \argmax_{a \in [n]} r_a - c_a$, 
while the optimal utility profile and the corresponding ESF is as follows (The graphical illustration is shown in Figure~\ref{fig:optimization_esf_risk_neutral} in Appendix~\ref{sec:opt_profile_esf}):
\begin{itemize}
	\item If $\frac{r_{a^*} - c_{a^*}}{2} \geq \max \left\{0, r_{\hat{a}} - c_{\hat{a}}\right\}$, the optimal utility profile is $\left(\frac{r_{a^*} - c_{a^*}}{2}, \frac{r_{a^*} - c_{a^*}}{2}\right)$ and the corresponding ESF is $\frac{r_{a^*} - c_{a^*}}{2}$.
	
	\item If the above condition is not satisfied, the optimal utility profile is $\left(r_{\hat{a}} - c_{\hat{a}}, r_{a^*} - c_{a^*} - r_{\hat{a}} + c_{\hat{a}}\right)$, and the corresponding ESF is $r_{a^*} - c_{a^*} - r_{\hat{a}} + c_{\hat{a}}$.
\end{itemize}

When the agent is \emph{risk-averse}, for a given implementable action $a$ with $r_a - v^{-1}(c_a) \geq \max \left\{0, r_{\hat{a}} - v^{-1}(c_{\hat{a}})\right\}$, denote $x_a$ as the solution of the equation $x = v(r_a - x) - c_a$, there are two possible scenarios:
\begin{itemize}
	\item $0 \leq x_a \leq r_{\hat{a}} - c_{\hat{a}}$, then for any implementable utility profiles in $F_a \cap \gF$, we have $\min \left\{u_a^P, u_a^A\right\} = u_a^A$. Therefore, maximizing ESF is equivalent to searching for the utility profile within $F_a \cap \gF$ whose $u_a^A$ is the maximum, which is $(r_{\hat{a}} - c_{\hat{a}}, v(r_a - r_{\hat{a}} + c_{\hat{a}}) - c_a)$ and the corresponding ESF is $v(r_a - r_{\hat{a}} + c_{\hat{a}}) - c_a$.
	
	\item $0 \leq r_{\hat{a}} - c_{\hat{a}} \leq x_a$. In this case, let $u_a^P = x, u_a^A = y$ the line $y = x$ would cross the implementable set $F_a \cap \gF$, and divide the set into two parts (see Figure~\ref{fig:optimization_esf_risk_averse} (a) in Appendix~\ref{sec:opt_profile_esf}): at the bottom right part, we have $\min \left\{u_a^P, u_a^A\right\} = u_a^A$, while at the top left part, we have $\min \left\{u_a^P, u_a^A\right\} = u_a^P$. Therefore, for the bottom right part, we aim to find the utility profile with the largest $u_a^A$. Since $y = x$ is strictly increasing and $y = v(r_a - x) - c_a$ is strictly decreasing, a single utility profile belongs to the bottom right part is optimal, i.e., $(x_a, x_a)$ where $x_a$ is the unique solution of the equation $x = v(r_a - x) - c_a$. For the top left part, the derivation is similar since both $y = x$ and $y = v(r_a - x) - c_a$ are strictly monotone functions, there is a one-to-one correspondence between $x$ and $y$. With similar calculations, we can see that the optimal utility profile for the top left part is the same as that for the bottom right part, which appears on the border of the two parts. Therefore, in this scenario, the optimal utility profile is $(x_a, x_a)$ with $x_a$ being the solution of the equation $x = v(r_a - x) - c_a$ and the corresponding ESF is $x_a$.
\end{itemize}

Consider two different implementable actions $a$ and $a'$, there are also three possible scenarios as in the risk-neutral case.
To illustrate, we first denote $x_a$ and $x_{a'}$ as the solution of the equations $x = v(r_a - x) - c_a$ and $x = v(r_{a'} - x) - c_{a'}$, respectively. Without loss of generality, we assume that $x_a \leq x_{a'}$. Furthermore, since $a$ and $a'$ are implementable actions, we have that $0 \leq x_a$.
\begin{itemize}
	\item $\max \left\{0, r_{\hat{a}} - v^{-1}(c_{\hat{a}})\right\} \leq x_a \leq x_{a'}$, both utility profiles $(x_a, x_a)$ and $(x_{a'}, x_{a'})$ are implementable, then $a'$ would be preferred since the ESF for the two actions are $x_a \leq x_{a'}$.
	
	\item $0 \leq x_a \leq r_{\hat{a}} - v^{-1}(c_{\hat{a}}) \leq x_{a'}$, the ESF for the two actions are $v(r_a - r_{\hat{a}} + v^{-1}(c_{\hat{a}})) - c_a$ and $x_{a'}$, respectively. Since $x_{a'} \geq x_a = v(r_a - x_a) - c_a \geq v(r_a - (r_{\hat{a}} - v^{-1}(c_{\hat{a}}))) - c_a$, where the last inequality follows from the fact that $v$ is a strictly increasing function, $a'$ is still the preferred action.
	
	\item $0 \leq x_a \leq x_{a'} \leq r_{\hat{a}} - v^{-1}(c_{\hat{a}})$, the optimal ESF for the two actions are $v(r_a - r_{\hat{a}} + v^{-1}(c_{\hat{a}})) - c_a$ and $v(r_{a'} - r_{\hat{a}} + v^{-1}(c_{\hat{a}})) - c_{a'}$, respectively. In this scenario, there is no single conclusion for the better action, and the choice depends on the exact value of $r_{\hat{a}} - v^{-1}(c_{\hat{a}})$.
	%	(see Figure for further illustration). 
	Therefore, we should directly compare the two optimal ESFs.
\end{itemize}
Based on the above observations, we can conclude that with a risk-averse agent, if the social planner uses ESF as the social utility function, the procedure for the determination of utility profile can be as follows:
\begin{enumerate}
	\item For each implementable action $a$, solve the equation $x = v(r_a - x) - c_a$ and take the solution as $x_a$.
	
	\item If $\exists a \in [n]$ such that $(x_a, x_a)$ is implementable, i.e., $x_a \geq \max \left\{0, r_{\hat{a}} - v^{-1}(c_{\hat{a}})\right\}$, then let $\gA = \left\{a: x_a \geq \max \left\{0, r_{\hat{a}} - v^{-1}(c_{\hat{a}}) \right\}\right\}$, and take $a^* = \argmax_{a \in \gA} x_a$, the optimal utility profile that maximizes ESF is $(x_{a^*}, x_{a^*})$ and the value of ESF is $x_{a^*}$.
	
	\item If $r_{\hat{a}} - v^{-1}(c_{\hat{a}}) \geq 0$ and $x_a < r_{\hat{a}} - v^{-1}(c_{\hat{a}}), \forall a \in [n]$, then $a^* = \argmax_{a \in [n]} v(r_a - r_{\hat{a}} + v^{-1}(c_{\hat{a}})) - c_a$, the optimal utility profile that maximized ESF is $(r_{\hat{a}} - v^{-1}(c_{\hat{a}}), v(r_{a^*} - r_{\hat{a}} + v^{-1}(c_{\hat{a}})) - c_{a^*})$ and the value of ESF is $v(r_{a^*} - r_{\hat{a}} + v^{-1}(c_{\hat{a}})) - c_{a^*}$.
\end{enumerate}

\paragraph{Approximated Fairness}
For approximated fairness, as is stated in Section~\ref{subsubsec:equality}, the ideal utility profile is achieved when $x = y$. Therefore, if the line $y = x$ has intersections with the implementable set, any utility profiles lie in the intersection is optimal in the sense that their AFs are all 0, i.e., the optimal utility profile would not be unique in this case. In particular, for risk-neutral agent, if $\gF_{AF} = \left\{(\frac{r_a - c_a}{2}, \frac{r_a - c_a}{2}): a \in [n], \frac{r_a - c_a}{2} \geq \max \left\{0, r_{\hat{a}} - c_{\hat{a}}\right\}\right\} \neq \emptyset$, then any $(x, y) \in \gF_{AF}$ is optimal. For risk-averse agent, the optimal utility profile set is
\begin{align*}
	\gF_{AF} = \left\{(x, y): a \in [n], y = x, x \geq \max \left\{r_{\hat{a}} - v^{-1}(c_{\hat{a}}), 0\right\}, y \leq v(r_a - x) - c_a\right\}
\end{align*}
if this set is non-empty.
On the other hand, if $\gF_{AF}$ is empty, i.e., the line $y = x$ has no intersection with the implementable set $\gF$, the optimal utility profile is unique, which is
\begin{align*}
	(r_{\hat{a}} - c_{\hat{a}}, r_{a^*} - c_{a^*} - r_{\hat{a}} + c_{\hat{a}}) \quad \text{with} \quad a^* = \argmax_{a \in [n]} r_a - c_a
\end{align*}
if the agent is risk-neutral, and is 
\begin{align*}
	(r_{\hat{a}} - v^{-1}(c_{\hat{a}}), v(r_{a^*} - r_{\hat{a}} + v^{-1}(c_{\hat{a}})) - c_{a^*}) \quad \text{with} \quad a^* = \argmax_{a \in [n]} v(r_a - r_{\hat{a}} + v^{-1}(c_{\hat{a}})) - c_a
\end{align*}
if the agent is risk-averse. For a graph illustration, see Figure~\ref{fig:optimization_af_risk_neutral} and \ref{fig:optimization_af_risk_averse} in Appendix~\ref{sec:opt_profile_af}.

Note that if the line $y = x$ has no intersection with the implementable set $\gF$, the two equality-based social utility function, egalitarian social welfare and approximated fairness, are equivalent in the sense that the optimal utility profile would be the same.

\section{Information Structure Design}\label{sec:info_struct_design}
Once the utility profile that maximizes the given social utility function is determined, the social planner should design an information structure such that it can induce an equilibrium for the Stackelberg game which leads to this utility profile.

We claim that no matter what risk attitude type of the agent, binary-signal information structure suffices for the design. 
Intuitively, designing an information structure is a way to compress the information of various actions. Since the principal designs the contract for each observable signal, actions belong to the same category of signals, i.e, $I_i = I_j$, yield the same expected transfer. 
Therefore, on the agent's side, only the cost $c_i$ affects the decision if two actions $i$ and $j$ share the same probability distribution over the signals.
Based on this observation, we can divide all the actions into two categories, one includes all the actions whose costs no less than the implemented action, while the other contains the remaining less costly actions.
In this way, the representative actions of the two categories are the desired action $a^*$ and $\hat{a}$, respectively. A ``high'' signal and a ``low'' signal suffice to distinguish these two actions and induce the equilibrium of the Stackelberg game towards the desired utility profile, whereas the probability mapping should be designed carefullly based on the risk attitude of the agent.

\paragraph{Risk-neutral agent}
When the agent is risk-neutral with utility function $u_a^A = t_a - c_a$, with a given utility file $(x^*, y^*)$ and its corresponding action-transfer pair $(a^*, s^*)$, the social planner design the information structure as follows:
\begin{align}\label{eq:info_struct_risk_neutral}
	I_i := \left\{
	\begin{array}{ll}
		(1, 0) & \text{if } i = a^* \text{ or } i \in \left[n\right] \ \st c_i > c_{a^*}, \\
		(p^*, 1 - p^*) & \text{otherwise,}
	\end{array}
	\right.
\end{align}
where $p^* := 1 - \frac{c_{a^*} - c_{\hat{a}}}{s^*}$.

\begin{thm}\label{thm:info_struct_risk_neutral}
	Suppose the principal and the agent are risk-neutral with utility functions given by $u_a^P = r_a - t_a$ and $u_a^A = t_a - c_a$, respectively. The information structure $I$ given by Eq.(\ref{eq:info_struct_risk_neutral}) induces the action-transfer pair $(a^*, s^*)$ and hence the utility profile $(x^*, y^*)$.
\end{thm}

\begin{proof}
	Firstly, under the designed information structure $I$, the agent only needs to consider a representative action for each category and hence the optimal action $i$ for a given contract must be one of $0, \hat{a}$, and $a^*$. Specifically, any $i$ with $c_i > c_{a^*}$ is classified in the same category as $a^*$ in the information structure, hence it would yield the same transfers as $a^*$ does, which leads to lower utility for the agent since $u_a^A = t_a - c_a$ and therefore is inferior to $a^*$. Furthermore, all the other actions are in the same category as action $\hat{a}$, while action $\hat{a}$ is the one with the lowest cost, and for the same reason as above, these actions are inferior and can be ignored.
	
	Next, we prove that with the information structure $I$ given by Eq.(\ref{eq:info_struct_risk_neutral}), the Stackelberg game equilibrium is the one with $(a^*, s^*)$ as the induced action and monetary transfer, and the optimal contract that leads to the equilibrium is $t^* = (t_1^*, t_2^*) = (s^*, 0)$.
	
	From the agent's perspective, we have that 
	\begin{align*}
		u_{0}^A &= 0, \\
		u_{\hat{a}}^A &= p^* \cdot s^* - c_{\hat{a}} = s^* - c_{a^*}, \\
		u_{a^*}^A &= s^* - c_{a^*},
	\end{align*}
	 which implies that $u_{\hat{a}}^A = u_{a^*}^A \geq u_{0}^A$. And from the principal's perspective, 
	 \begin{align*}
	 	u_{0}^P &= 0, \\
	 	u_{\hat{a}}^P &= r_{\hat{a}} - p^* \cdot s^* = r_{\hat{a}} - s^* + c_{a^*} - c_{\hat{a}}, \\
	 	u_{a^*}^P &= r_{a^*} - s^*.
	 \end{align*}
	 Therefore, $u_{\hat{a}}^P - u_{a^*}^P = (r_{\hat{a}} - c_{\hat{a}}) - (r_{a^*} - c_{a^*}) \leq 0$, where the inequality derives from the implementabiltiy of the action $a^*$.
	 By the tie-breaking rule, the agent would choose action $a^*$ which is in the favor of the principal, and the principal's utility is $r_{a^*} - s^*$. In the following, we show that this is the maximal utility that the principal can ensure and hence $t^*$ is the optimal contract.
	
	On one hand, any contract that produces action 0 or $\hat{a}$ would not be better than the proposed contract $t^*$. To see this, note that if the principal releases a contract that induces action 0, then she will receive a utility of 0, while by the selection rule of $s^*$, we have $r_{a^*} - s^* \geq 0$, which implies that the above contract is no better than $t^*$. In addition, if the principal releases a contract that induces action $\hat{a}$, then she must have an expected transfer $t_1 \geq c_{\hat{a}}$, otherwise the agent would have negative utility under action $\hat{a}$ and chooses action 0 instead. The principal would get a utility of $r_{\hat{a}} - t_1 \leq r_{\hat{a}} - c_{\hat{a}} = w_1 \leq r_{a^*} - s^*$, where the last inequality stems from the selection of $s^*$. Therefore, these contracts are inferior to the contract $t^*$. 
	
	On the other hand, if the principal releases a contract $\tilde{t} = (\tilde{t}_1, \tilde{t}_2)$ inducing action $a^*$ but with different transfers from $t^*$, then $\tilde{t}$ must satisfy 
	\begin{align*}
		\begin{split}
			\tilde{t}_1 - c_{a^*} &\geq \tilde{t}_1 - c_{\hat{a}}, \\
			\tilde{t}_1 - c_{a^*} &\geq p \tilde{t}_1 + (1 - p) \tilde{t}_2 - c_{\hat{a}}, 
		\end{split}
	\end{align*}
	otherwise, the agent would choose action $\hat{a}$ instead. As $p \in \left[0, 1\right]$ and by our assumption that $\tilde{t}_j \geq 0$, we have $\tilde{t}_1 - c_{a^*} \geq p \tilde{t}_1 - c_{\hat{a}}$, which implies that $\tilde{t}_1 \geq \frac{c_{a^*} - c_{\hat{a}}}{1 - p} = s^*$. Therefore, under the information structure $I$ determined by the social planner, $s^*$ is the minimum transfer that the principal could set to induce the action $a^*$, and hence $t^*$ is the optimal contract as desired.
\end{proof}

\paragraph{Risk-averse agent}
When the agent is risk-averse, i.e., the von Neumann-Morgenstern utility $v: \sR_{\geq 0} \rightarrow \sR_{\geq 0}$ of the agent is a concave function, and $v$ is twice differentiable~\citep{nielsen1999differentiable, nakamura2015differentiability}, strictly increasing with $v(0) = 0$ and $\lim_{z \rightarrow \infty} \frac{v(z)}{z} = 0$. 
Let $z^* \in \sR_{\geq 0}$ be the solution to the equation $\frac{v(z)}{z} = \frac{y + c_{a^*}}{r_{a^*} - x}$.
Given an implementable utility profile $(x^*, y^*)$ and its corresponding action-transfer pair $(a^*, s^*)$, set 
\begin{align*}
	p^* := \frac{r_{a^*} - x^*}{z^*} \frac{y^* + c_{\hat{a}}}{y^* + c_{a^*}} \quad \text{and} \quad q^* := \frac{r_{a^*} - x^*}{z^*}.
\end{align*}
the social planner design the information structure as follows.:
\begin{align}\label{eq:info_struct_risk_averse}
	I_i := \left\{
	\begin{array}{ll}
		(q^*, 1 - q^*) & \text{if } i = a^* \text{ or } i \in \left[n\right] \ \st c_i > c_{a^*}, \\
		(p^*, 1 - p^*) & \text{otherwise.}
	\end{array}
	\right.
\end{align}

\begin{thm}\label{thm:info_struct_risk_averse}
	Suppose the principal and the agent are risk-neutral and risk-averse, respectively. The information structure $I$ given by Eq.(\ref{eq:info_struct_risk_averse}) induces the action-transfer pair $(a^*, s^*)$ and hence the utility profile $(x^*, y^*)$.
\end{thm}

\begin{proof}
	Firstly, we prove that the equation $\frac{v(z)}{z} = \frac{y + c_{a^*}}{r_{a^*} - x}$ admits a unique solution and $p^*, q^* \in \left[0, 1\right]$ so that the information structure in Eq.(\ref{eq:info_struct_risk_averse}) is well-defined.
	
	On one hand, consider the function $f(z) = \frac{v(z)}{z}$, we have $f'(z) = \frac{v'(z)z - v(z)}{z^2}$. Let $g(z) = v'(z)z - v(z)$, then $g'(z) = v''(z) z$. By the fact that $v(\cdot)$ is a strictly concave function and it is defined on $\sR_{\geq 0}$, we have $g'(z) \leq 0$ with the equality holds when $z = 0$ and hence $g(z) < g(0) = 0$ for any $z < 0$. Therefore, $f'(z) \leq 0$ with the equality holds when $z = 0$ and $f(z)$ is strictly monotonically decreasing on $\sR_+$.
	
	On the other hand, denote the right-hand-side derivative of $v$ at 0 as $d_0 := v_+'(0)$, then by L'H\^{o}pital's rule we have $\lim_{z \rightarrow 0+} \frac{v(z)}{z} = d_0$ and by assumption $\lim_{z \rightarrow \infty} \frac{v(z)}{z} = 0$. Since $v(\cdot)$ is strictly concave, we have 
	\begin{align}\label{eq:concavity_v}
		v(r_{a^*} - x^*) \leq d_0 \cdot (r_{a^*} - x^*).
	\end{align}
	From the characterization of implementable utility profile, we have that 
	\begin{align}\label{eq:y_implementability}
		y^* \leq v(r_{a^*} - x^*) - c_{a^*}.
	\end{align}
	Combining Eq. (\ref{eq:concavity_v}) and (\ref{eq:y_implementability}) and the characterization of the implementable utility profile, we have 
	\begin{align*}
		0 \leq \frac{y^* + c_{a^*}}{r_{a^*} - x^*} \leq d_0.
	\end{align*}
	By the intermediate value theorem and the monotonicity of $v(\cdot)$, we have that the equation $\frac{v(z)}{z} = \frac{y^* + c_{a^*}}{r_{a^*} - x^*}$ admits a solution and further the solution is unique.
	
	By the implementability of $(x^*, y^*)$ and $c_{a^*} \geq c_{\hat{a}}$, we have $0 \leq p^* \leq q^*$. 
	We prove $q^* \leq 1$ by contradiction, i.e., assume $r_a - x^* > z^*$, then by the monotonicity of $\frac{v(z)}{z}$ and Eq. (\ref{eq:y_implementability}), we have
	\begin{align*}
		\frac{v(z^*)}{z^*} > \frac{v(r_{a^*} - x^*)}{r_{a^*} - x^*} \geq \frac{y^* + c_{a^*}}{r_{a^*} - x^*},
	\end{align*}
	which contradicts with the fact that $\frac{v(z^*)}{z^*} = \frac{y^* + c_{a^*}}{r_{a^*} - x^*}$.
	
	Secondly, similar to the case with a risk-neutral agent, any actions $i \notin \left\{0, \hat{a}, a^*\right\}$ are suboptimal for the risk-averse agent under any contract, since action $\hat{a}$ and action $a^*$ are the least costly actions among their categories, respectively, and the distribution over transfers are the same within the same information category. Therefore, we can focus on the actions $\left\{0, \hat{a}, a^*\right\}$.
	
	Next, we prove that with the information structure $I$ defined as in Eq. (\ref{eq:info_struct_risk_averse}), the optimal contract that leads to the equilibrium is $t^* = (t_1^*, t_2^*) = (z^*, 0)$.
	
	From the agent's perspective, her expected utility with actions 0, $\hat{a}$ and $a^*$ are
	\begin{align*}
		u_{0}^{A} &= 0, \\
		u_{\hat{a}}^{A} &= p^* \cdot v(z^*) - c_{\hat{a}} = \frac{r_{a^*} - x^*}{z^*} \frac{y^* + c_{\hat{a}}}{y^* + c_{a^*}} \cdot v(z^*) - c_{\hat{a}} = \frac{(r_{a^*} - x^*) \cdot (y^* + c_{\hat{a}})}{y^* + c_{a^*}} \frac{y^* + c_{a^*}}{r_{a^*} - x^*} - c_{\hat{a}} = y^*, \\
		u_{a^*}^{A} &= q^* \cdot v(z^*) - c_{a^*} = \frac{r_{a^*} - x^*}{z^*} \cdot v(z^*) - c_{a^*} = (r_{a^*} - x^*) \frac{y^* + c_{a^*}}{r_{a^*} - x^*} - c_{a^*} = y^*, 
	\end{align*}
	where the second equality for the calculation of $u_{\hat{a}}^A$ and $u_{a^*}^A$ are from the definition of $z^*$.
	From the above calculation, we find that $u_{\hat{a}}^A = u_{a^*}^A \geq u_{0}^A$, 
	By the tie-breaking rule, the agent would choose the action in the principal's favor and hence chooses action $a^*$. 
	Therefore, the contract $t^* = (z^*, 0)$ induces the agent's action $a^*$, and we need to show that it is optimal from the principal's perspective.
	
	On one hand, any contract that induces action 0 or $\hat{a}$ would be dominated by the proposed contract $t^*$. By the implementability of the utility profile $(x^*, y^*)$, we have $x^* \geq 0$, which implies that the principal would prefer contract $t^*$ that provides her with a utility of $x^*$ instead of 0 from action 0. Additionally, among all the contracts inducing action $\hat{a}$, the optimal one is $(v^{-1}(c_{\hat{a}}), v^{-1}(c_{\hat{a}}))$, which offers the agent the minimal subsidy while maintaining the incentive.
	Under such contract, the principal receives a utility of $r_{\hat{a}} - v^{-1}(c_{\hat{a}}) \leq x^*$, where the inequality comes from the implementability of $x^*$ and hence the principal would prefer the contract $t^*$.
	
	On the other hand, consider a contract $\tilde{t} = (\tilde{t}_1, \tilde{t}_2)$ inducing action $a^*$ but with different transfers from $t^*$, then $\tilde{t}$ must satisfy
	\begin{align*}
		u_{a^*}^A(\tilde{t}) \geq u_1^A(\tilde{t}) \quad \text{and} \quad u_{a^*}^A(\tilde{t}) \geq 0,
	\end{align*}
	otherwise the agent would choose action $\hat{a}$ or 0. Actually, the second condition can be withdrawn as it can be derived from the first one. To see this, assume that $u_{a^*}^A(\tilde{t}) \geq u_{\hat{a}}^A(\tilde{t})$, then by the definition of $u_{a^*}^A(\tilde{t})$ and $u_{\hat{a}}^A(\tilde{t})$, we have the following equivalent inequality
	\begin{align*}
		q^* \cdot v(\tilde{t}_1) + (1 - q^*) \cdot v(\tilde{t}_2) - c_{a^*} \geq p^* \cdot v(\tilde{t}_1) + (1 - p^*) \cdot v(\tilde{t}_2) - c_{\hat{a}}.
	\end{align*}
	By rearrangement, we further have the following equivalent expression
	\begin{align}\label{eq:value_ineq1}
		v(\tilde{t}_1)  \geq v(\tilde{t}_2) + \frac{c_{a^*} - c_{\hat{a}}}{q^* - p^*}
	\end{align}
	Then we have
	\begin{align}\label{eq:value_ineq2}
		\begin{split}
			u_{a^*}^A(\tilde{t}) &= q^* \cdot v(\tilde{t}_1) + (1 - q^*) \cdot v(\tilde{t}_2) - c_{a^*} \\
			&\geq q^* \cdot \left(v(\tilde{t}_2) + \frac{c_{a^*} - c_{\hat{a}}}{q^* - p^*}\right) + (1 - q^*) \cdot v(\tilde{t}_2) - c_{a^*} \\
			&= v(\tilde{t}_2) + \frac{p^*}{q^* - p^*} c_{a^*} + \frac{q^*}{q^* - p^*} c_{\hat{a}} \\
			&= v(\tilde{t}_2) + \frac{y^* + c_{\hat{a}}}{c_{a^*} - c_{\hat{a}}} c_{a^*} - \frac{y^* + c_{a^*}}{c_{a^*} - c_{\hat{a}}} \\
			&= v(\tilde{t}_2) + y^* \geq 0,
		\end{split}
	\end{align}
	where the first inequality follows from Eq.(\ref{eq:value_ineq1}) and the last inequality follows from the fact that $v$ is a nonnegative function and the implementability of $y^*$ implies $y^* \geq 0$.
	
	Therefore, the only requirement for $\tilde{t}$ is Eq.(\ref{eq:value_ineq1}). If $\tilde{t}_2 > 0$, slightly reducing $\tilde{t}_2$ preserves the inequality since $v$ is strictly increasing, which implies that a contract with $\tilde{t}_2 > 0$ is suboptimal for the principal. 
	Therefore, we only need to consider $\tilde{t} = (\tilde{t}_1, 0)$.
	
	If $u_{a^*}^A(\tilde{t}) > u_{1}^A(\tilde{t})$, the inequalities in Eq.(\ref{eq:value_ineq1}) and (\ref{eq:value_ineq2}) are strict and hence the principal can reduce the value of $\tilde{t}_1$ to some extent while still preserve the inequality. This implies that we must have $u_{a^*}^A(\tilde{t}) = u_{\hat{a}}^A(\tilde{t})$, i.e., $v(\tilde{t}_1) = \frac{c_{a^*} - c_{\hat{a}}}{q^* - p^*}$ as $\tilde{t}_2 = 0$ and $v(0) = 0$. Substitute the value of $p^*$ and $q^*$ into the expression of $v(\tilde{t}_1)$, we have
	\begin{align*}
		v(\tilde{t}_1) = (c_{a^*} - c_{\hat{a}}) \cdot \frac{(y + c_{a^*}) \cdot z^*}{(r_{a^*} - x^*) \cdot (c_{a^*} - c_{\hat{a}})} = z^* \cdot \frac{y + c_{a^*}}{r_{a^*} - x^*},
	\end{align*}
	which implies that $v(\tilde{t}_1) = v(z^*)$. Since $v$ is strictly increasing, we must have $\tilde{t}_1 = z^*$, the unique solution to the equation $\frac{v(z)}{z} = \frac{y + c_{a^*}}{r_{a^*} - x^*}$. In other words, $t^* = (z^*, 0)$ is the unique optimal contract that under the information structure $I$ as given by Eq.(\ref{eq:info_struct_risk_averse}).
\end{proof}

\section{Extension to Risk-Averse Principal}\label{sec:risk_averse_principal}
If the principal is risk-averse with utility function $u(x)$ over the actual income $x = r_k - t_j$, then the expected utility function for the principal is 
\begin{align}\label{eq:utility_risk_averse_principal}
	u_a^P = \expect_{k \sim P_a} \expect_{j \sim I_a} u(r_k - t_j), 
\end{align}
where $u(\cdot)$ is a concave, increasing function with $u(0) = 0$.

For simplicity, we consider the agent is risk neutral over her income and hence the expected utility function is $u_a^A = \expect_{j \sim I_a} t_j - c_a := t_a - c_a$.
By the concavity of function $u(\cdot)$, we have
\begin{align*}
	\begin{split}
		u_a^P &= \expect_{k \sim P_a} \expect_{j \sim I_a} u(r_k - t_j) \\
		&\leq u \left(\expect_{k \sim P_a} \expect_{j \sim I_a} r_k - t_j\right) \\
		&= u \left(r_a - u_a^A - c_a\right).
	\end{split}
\end{align*}
By the monotonicity of $u(\cdot)$, we have $u_a^A \leq -u^{-1}(u_a^P) + r_a - c_a$.
Graphical representation of the relationship between $u_a^P$ and $u_a^A$ is shown in Figure~\ref{fig:extension_risk_averse} in Appendix~\ref{sec:profile_risk_averse_principal}, which is similar as that of a risk-neutral principal and risk-averse agent. 
Therefore, the methodology in Section~\ref{sec:utility_profile_determination} and \ref{sec:info_struct_design} can be naturally extended to the case with a risk-averse principal.

\section{Conclusion}\label{sec:discuss}
In this paper, we have considered the principal-agent problem with a third party that we called the social planner. The social planner can control the information flow between the principal and the agent, and aims to induce socially optimal outcome for the system.
We have devised a workflow for the social planner. First, with a specific social utility function, the social planner solves a constrained optimization problem to determine the optimal utility profile. Because the social planner is faced with a system consisting of only two players, in most cases the optimization problem can be easily solved by a geometric approach, i.e., by graphing the utility function of the players to figure out the implementable sets and then using the contour line of the objective social utility function to determine the optimal utility profile.
Secondly, having the optimal utility profile in mind, we have provided a simple binary-signal information structure for the cases where agent with different risk attitudes, i.e., risk-neutral and risk-averse. Under the designed information structure, the Stackelberg game is guaranteed to arrive at an equilibrium of the desired utility profile, which in turns would maximize the social utility function and satisfy the purpose of the social planner. 
Our method is simple and modular, and the result shows that information plays a key role in social planning in the principal-agent model.

\bibliographystyle{plainnat}
\bibliography{reference}

\appendix

\section{Calculation of Utility Profile for Nash Product with a Risk-Averse Agent}\label{sec:np_risk_averse}
In this section, we prove that Eq. (\ref{eq:utility_profile_np_risk_averse}) provides the maximum Nash product for any $(u_a^P, u_a^A) \in F_a \cap \gF$.

Firstly, without loss of generality, we assume that action $a$ is implementable (otherwise $F_a \cap \gF = \emptyset$), then we have $r_a - v^{-1}(c_a) \geq \max \left\{0, r_{\hat{a}} - v^{-1}(c_{\hat{a}})\right\}$.
On one hand, as is stated in Section~\ref{subsubsec:overall_optimization}, the utility profile that attains the maximum Nash product could only appear on the boundary of the implementable set, i.e., $y = v(r_a - x) - c_a$, we can focus on this function. 
On the other hand, if a utility profile $(u^P, u^A)$ achieves a Nash product $z$, then the relationship between $u^A$ and $u^P$ can be described by a function $y = \frac{z}{x}$.

Intuitively speaking, by changing the value of $z$, the equation $v(r_a - x) - c_a = \frac{z}{x}$ would have 0, 1 or 2 roots. When the equation has unique root, the corresponding $z$ is the maximum Nash product, and the root is the utility of the principal at the equilibrium.
In particular, when $z$ is given, let
\begin{align*}
	f(x) = \frac{z}{x} - v(r_a - x) + c_a,
\end{align*}
then $f(x) \leq 0$ implies that there exists an implementable utility profile where the agent's action at the equilibrium is $a$ and the Nash product of this utility profile is $z$. 
We can observe the following two facts:
\begin{enumerate}
	\item $f'(x) = - \frac{z}{x^2} + v'(r_a - x)$ and $f''(x) = \frac{2z}{x^3} - v''(r_a - x) \geq 0$ by the concavity of $v$, which implies that $f(x)$ is a convex function. Furthermore, when $x \to 0$, we have $f(x) \to +\infty$ and $f'(x) \to -\infty$, and when $x_a = r_a - v^{-1}(c_a)$, we have $f(x_a) = \frac{z}{x_a} > 0$, $f'(x_a) = - \frac{z}{x_a^2} + v'(v^{-1}(c_a))$. If $z$ is sufficiently small, we would have $f'(x_a) \geq 0$ and it further implies that $f(x)$ is unimodal with a unique minimum in the interval $\left[0, r_a - v^{-1}(c_a)\right]$. And $\min_{x \in \left[0, r_a - v^{-1}(c_a)\right]} f(x) >0, =0, <0$ correspond to the case that the equation $v(r_a - x) - c_a = \frac{z}{x}$ has 0, 1, 2 roots.
	
	\item For every given $x > 0$, $\frac{z}{x} - v(r_a - x) + c_a$ is a strictly increasing linear function in terms of $z$, and hence increasing $z$ would uniformly increase the value of the function for every $x > 0$, i.e., the function curve would shift to the up. As a result, as $z$ approaches 0, the number of solutions of the equation would change from 0 to 1 to 2.
\end{enumerate}
Based on the above two facts, it suffices to prove that there exists $z$, such that the equation $\frac{z}{x} - v(r_a - x) + c_a$ has only one root, the corresponding $z$ would be the maximum attainable Nash product, if the root $x \in \left[\max \left\{0, r_{\hat{a}} - v^{-1}(c_{\hat{a}})\right\}, r_a - v^{-1}(c_a)\right]$. In the following Lemma~\ref{lem:np_max}, we prove this claim.

\begin{lem}\label{lem:np_max}
	Let $v: \sR_{\geq 0} \rightarrow \sR_{\geq 0}$ be a strictly increasing concave function that satisfies $v(0) = 0$ and $\lim_{z \to \infty} \frac{v(z)}{z} = 0$, 
	and let $r$ and $c$ be a pair of given nonnegative real numbers satisfy $r - v^{-1}(c) \geq 0$, 
	there exists $z \geq 0$, such that the equation $\frac{z}{x} = v(r - x) - c$ has only one root on the interval $0 \leq x \leq r - v^{-1}(c)$.
\end{lem}

\begin{proof}
	We prove in the following that the conclusion is valid on the interval $0 \leq x \leq r$.
	If $\frac{z}{x} = v(r - x) - c$, then we have
	\begin{align}\label{eq:np_max_eq1}
		x \cdot v(r - x) - cx = z \quad \Rightarrow \quad x \cdot v(r - x) = cx + z.
	\end{align}
	For LHS of Eq.(\ref{eq:np_max_eq1}), let $f(x) = x \cdot v(r - x)$, then by the fact that $v$ is strictly increasing and $v(0) = 0$, we have that
	\begin{align*}
		f(x) > 0, &\quad x \in (0, r), \\
		f(x) = 0, &\quad x = 0 \text{ or } r.
	\end{align*}
	For the first order derivative $f'(x) = v(r - x) - x \cdot v'(r - x)$, 
	\begin{align*}
		f'(x) = \left\{
		\begin{array}{ll}
			v(r) > 0 & x = 0, \\
			-r \cdot v'(0) < 0 & x = r.
		\end{array}
		\right.
	\end{align*}
	Furthermore, for the second derivative, 
	\begin{align*}
		f''(x) &= -v'(r - x) - \left[v'(r - x) - x \cdot v''(r - x)\right] \\
		&= -2 v'(r - x) + x \cdot v''(r - x) < 0 \quad \text{if } x \in [0, r], 
	\end{align*}
	since $v$ is strictly increasing and concave. Therefore, $f'(x)$ is strictly decreasing on $[0, r]$ and there is a unique $x$ such that $f'(x) = 0$, which implies that $f(x)$ is concave and increases then decreases on $[0, r]$.
	
	For RHS of Eq.(\ref{eq:np_max_eq1}), denote $g(x) = cx + z$, then $g(x)$ is a linear function of $x$, whose slope is $c$. Therefore, to get the conclusion of this lemma, it suffices to prove that there exists $z > 0$, such that $g(x)$ is a tangent line to $f(x)$.
	
	From the assumption, we know that $r - v^{-1}(c) \geq 0$, i.e., $v(r) \geq c$. While on $[0, r]$, $f'(x)$ strictly decreasing from $v(r)$ to $-r \cdot v'(0)$, i.e., $v(r) \geq c \geq -r \cdot v'(0)$, then there exists some $x$ in $[0, r]$ such that $f'(x) = c$. Denote the corresponding $x$ as $x_c$, then let $z = f(x_c) - c \cdot x_c$, we have $g(x)$ is tangent to $f(x)$ at $x_c$. Furthermore, since $f(0) = 0, f'(0) = v(r) \geq c$, we have that the intercept $z \geq 0$.
\end{proof}

From the proof of Lemma~\ref{lem:np_max}, the calculation of the utility profile comes out naturally.
Firstly, we solve the equation $v(r_a - x) - x \cdot v'(r_a - x) = c_a$ and get the solution $x_a$ and let $y_a = v(r_a - x_c) - c_a$.
\begin{itemize}
	\item If $x_a \in \left[\max \left\{0, r_{\hat{a}} - v^{-1}(c_{\hat{a}})\right\}, r_a - v^{-1}(c_{\hat{a}})\right]$,
	the utility profile is $(u_a^P, u_a^A) = (x_a, y_a)$ and the maximum Nash product that is attainable for action $a$ is $x_a \cdot y_a$.
	
	\item If $r_{\hat{a}} - v^{-1}(c_{\hat{a}}) \leq 0$, then by Lemma~\ref{lem:np_max}, we know that $x_a \geq 0$ and the utility profile $(x_a, y_a)$ could be achieved. 
	
	\item However, if $r_{\hat{a}} - v^{-1}(c_{\hat{a}}) > 0$ and $0 \leq x_a \leq r_{\hat{a}} - v^{-1}(c_{\hat{a}})$, we have that $(x_a, y_a) \notin F_a \cap \gF$. In this case, we have to decrease $z$ such that there is some point on the curve $y = \frac{z}{x}$ lies in the implementable set. Denote $x_1 = r_{\hat{a}} - v^{-1}(c_{\hat{a}})$, the maximum $z$ such that $y = \frac{z}{x} \cap \left(F_a \cap \gF\right) \neq \emptyset$ is $x_1 \cdot \left(v(r_a - x_1) - c_a\right)$ and the corresponding utility profile is $(u_a^P, u_a^A) = (x_1, v(r_a - x_1) - c_a)$.
\end{itemize}

\section{Optimal Utility Profile for a Given Action with Utilitarian Social Welfare and Risk-Averse Agent}\label{sec:opt_profile_usf_averse}
Figure~\ref{fig:optimization_usf_risk_averse} illustrates the optimal utility profile for a given action $a$ if the social planner uses utilitarian social welfare (USF) as the social utility function and the agent is risk-averse.

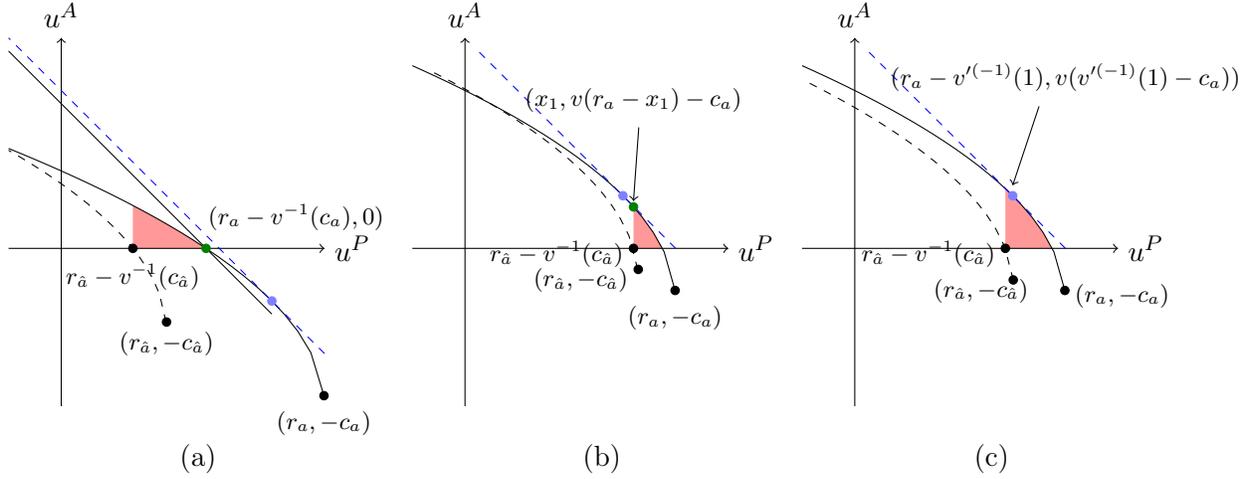
\begin{figure}[!htp]
	\centering
	\begin{tikzpicture}[xscale=0.7, yscale=1.4]
		\fill [red!40, domain = 2.251:3.64, variable = \x]
		(1.36, 0)--
		plot ({-\x+5}, {sqrt{\x} - 1.5})--cycle;
		
		\draw[->] (-1, 0)--(5,0);
		\draw[->] (0, -1.5)--(0,2);
		\node[right] at (5,0) {$u^P$};
		\node[above] at (0,2) {$u^A$};
		
		\draw[dashed, domain = 0.01:3, variable = \x]
		plot ({-\x+2}, {sqrt{\x} - 0.8});
		
		\draw[domain = 0.01:6, variable = \x]
		plot ({-\x+5}, {sqrt{\x} - 1.5});
		
		\filldraw (2, -0.7) ellipse (0.08 and 0.04);
		\node[below] at (2, -0.7) {\footnotesize $(r_{\hat{a}}, -c_{\hat{a}})$};
		
		\filldraw (4.99, -1.4) ellipse (0.08 and 0.04);
		\node[below] at (4.99, -1.45) {\footnotesize $(r_a, -c_a)$};
		
		\filldraw (1.36, 0) ellipse (0.08 and 0.04);
		\node[below] at (1.36, -0.05) {\footnotesize $r_{\hat{a}} - v^{-1}(c_{\hat{a}})$};
		
		\draw (4,-0.625)--(-1, 1.875);
		\draw [dashed, color = blue] (5,-1)--(-1,2);
		
		\filldraw [color = green!50!black, fill = green!50!black] (2.75, 0) ellipse (0.08 and 0.04);
		\node[right] at (2.6, 0.3) {\footnotesize $(r_a - v^{-1}(c_a), 0)$};
		
		\filldraw [color = blue!50, fill = blue!50] (4, -0.5) ellipse (0.08 and 0.04);
		
		\node at (2.6, -2) {(a)};
	\end{tikzpicture}
	\hspace{0mm}
	\begin{tikzpicture}[xscale=0.7, yscale=1.4]
		\fill [red!40, domain = 0.25:0.8, variable = \x]
		(3.2, 0)--
		plot ({-\x+4}, {sqrt{\x} - 0.5})--cycle;
		
		\draw[->] (-1, 0)--(5,0);
		\draw[->] (0, -1.5)--(0,2);
		\node[right] at (5,0) {$u^P$};
		\node[above] at (0,2) {$u^A$};
		
		\draw[dashed, domain = 0.01:4, variable = \x]
		plot ({-\x+3.29}, {sqrt{\x} - 0.3});
		
		\draw[domain = 0.01:5, variable = \x]
		plot ({-\x+4}, {sqrt{\x} - 0.5});
		
		\filldraw (3.29, -0.2) ellipse (0.08 and 0.04);
		\node[left] at (3.29, -0.3) {\footnotesize $(r_{\hat{a}}, -c_{\hat{a}})$};
		
		\filldraw (3.99, -0.4) ellipse (0.08 and 0.04);
		\node[below] at (3.99, -0.45) {\footnotesize $(r_a, -c_a)$};
		
		\filldraw (3.2, 0) ellipse (0.08 and 0.04);
		\node[left] at (3.2, -0.05) {\footnotesize $r_{\hat{a}} - v^{-1}(c_{\hat{a}})$};
		
		\draw [dashed, color = blue] (4,0)--(0.2,1.9);
		
		\filldraw[color = blue!50, fill = blue!50] (3,0.5) ellipse (0.08 and 0.04);
		
		\filldraw [color = green!50!black, fill = green!50!black] (3.2, 0.394) ellipse (0.08 and 0.04);
		\node [above] at (3.2, 1.2) {\footnotesize $(x_1, v(r_a - x_1) - c_a)$};
		\draw[->] (3.3, 1.15)--(3.2, 0.45);
		
		%		\filldraw[color = green!50!black, fill = green!50!black] (2.86, 0.463) ellipse (0.06 and 0.03);
		
		\node at (2.6, -2) {(b)};
	\end{tikzpicture}
	\hspace{0mm}
	\begin{tikzpicture}[xscale=0.7, yscale=1.4]
		\fill [red!40, domain = 0.25:1.14, variable = \x]
		(2.86, 0)--
		plot ({-\x+4}, {sqrt{\x} - 0.5})--cycle;
		
		\draw[->] (-1, 0)--(5,0);
		\draw[->] (0, -1.5)--(0,2);
		\node[right] at (5,0) {$u^P$};
		\node[above] at (0,2) {$u^A$};
		
		\draw[dashed, domain = 0.01:4, variable = \x]
		plot ({-\x+3.02}, {sqrt{\x} - 0.4});
		
		\draw[domain = 0.01:5, variable = \x]
		plot ({-\x+4}, {sqrt{\x} - 0.5});
		
		\filldraw (3.01, -0.3) ellipse (0.08 and 0.04);
		\node[left] at (3.49, -0.4) {\footnotesize $(r_{\hat{a}}, -c_{\hat{a}})$};
		
		\filldraw (3.99, -0.4) ellipse (0.08 and 0.04);
		\node[right] at (3.99, -0.45) {\footnotesize $(r_a, -c_a)$};
		
		\filldraw (2.86, 0) ellipse (0.08 and 0.04);
		\node[left] at (2.86, -0.05) {\footnotesize $r_{\hat{a}} - v^{-1}(c_{\hat{a}})$};
		
		\draw [dashed, color = blue] (4,0)--(0.2,1.9);
		
		\filldraw[color = blue!50, fill = blue!50] (3,0.5) ellipse (0.08 and 0.04);
		
		\node[above] at (4, 1.4) {\footnotesize $(r_a - v'^{(-1)}(1), v(v'^{(-1)}(1) - c_a))$};
		\draw[->] (3.5,1.35)--(3,0.6);
		
		\node at (2.6, -2) {(c)};
	\end{tikzpicture}
	\caption{Optimal utility profile for a given action $a$ if the social planner uses utilitarian social welfare (USF) as the social utility function and the agent is risk-averse. Let $x_1 = r_{\hat{a}} - v^{-1}(c_{\hat{a}})$. (a) If $v'(v^{-1}(c_a)) \leq 1$, the green point represents the optimal utility profile. (b) If $r_{\hat{a}} - v^{-1}(c_{\hat{a}}) \geq 0$ and $v'(r_a - x_1) > 1$, the green point represents the optimal utility profile. (c) If the line $y = x$ is tangent to $y = v(r_a - x) - c_a$ at some implementable utility profile, the blue point represents the optimal utility profile.}
	\label{fig:optimization_usf_risk_averse}
\end{figure}

\section{Optimal Utility Profile for a Given Action with Nash Product and Risk-Averse Agent}\label{sec:opt_profile_np_averse}
Figure~\ref{fig:optimization_np_risk_averse} illustrates the optimal utility profile for a given action $a$ if the social planner uses Nash product as the social utility function and the agent is risk-averse.

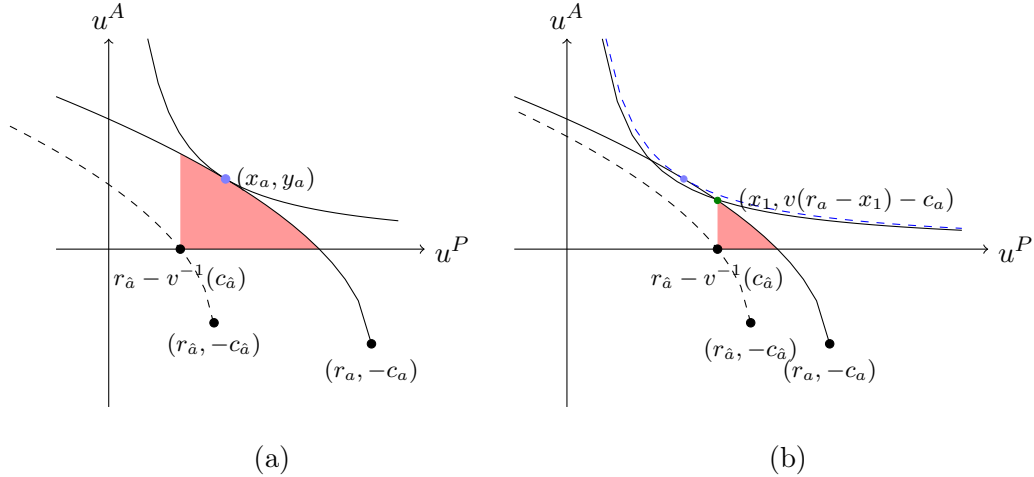
\begin{figure}[!htp]
	\centering
	\begin{tikzpicture}[xscale=0.7, yscale=1.4]
		\fill [red!40, domain = 1.001:3.64, variable = \x]
		(1.36, 0)--
		plot ({-\x+5}, {sqrt{\x} - 1})--cycle;
		
		\draw[->] (-1, 0)--(6,0);
		\draw[->] (0, -1.5)--(0,2);
		\node[right] at (6,0) {$u^P$};
		\node[above] at (0,2) {$u^A$};
		
		\draw[dashed, domain = 0.01:4, variable = \x]
		plot ({-\x+2}, {sqrt{\x} - 0.8});
		
		\draw[domain = 0.01:6, variable = \x]
		plot ({-\x+5}, {sqrt{\x} - 1});
		
		\filldraw (2, -0.7) ellipse (0.08 and 0.04);
		\node[below] at (2, -0.7) {\footnotesize $(r_{\hat{a}}, -c_{\hat{a}})$};
		
		\filldraw (4.99, -0.9) ellipse (0.08 and 0.04);
		\node[below] at (4.99, -0.95) {\footnotesize $(r_a, -c_a)$};
		
		\filldraw (1.36, 0) ellipse (0.08 and 0.04);
		\node[below] at (1.36, -0.05) {\footnotesize $r_{\hat{a}} - v^{-1}(c_{\hat{a}})$};
		
		\draw [domain = (20/27):5.5, variable = \x]
		plot ({\x}, {40 / (27*\x)});
		
		\filldraw[color = blue!50, fill = blue!50] (2.22, 0.667) ellipse (0.08 and 0.04);
		\node[right] at (2.22, 0.667) {\footnotesize $(x_a, y_a)$};
		
		\node at (3.1, -2) {(a)};
	\end{tikzpicture}
	\hspace{2mm}
	\begin{tikzpicture}[xscale=0.7, yscale=1.4]
		\fill [red!40, domain = 1.001:2.14, variable = \x]
		(2.86, 0)--
		plot ({-\x+5}, {sqrt{\x} - 1})--cycle;
		
		\draw[->] (-1, 0)--(8,0);
		\draw[->] (0, -1.5)--(0,2);
		\node[right] at (8,0) {$u^P$};
		\node[above] at (0,2) {$u^A$};
		
		\draw[dashed, domain = 0.01:4.5, variable = \x]
		plot ({-\x+3.5}, {sqrt{\x} - 0.8});
		
		\draw[domain = 0.01:6, variable = \x]
		plot ({-\x+5}, {sqrt{\x} - 1});
		
		\filldraw (3.49, -0.7) ellipse (0.08 and 0.04);
		\node[below] at (3.49, -0.75) {\footnotesize $(r_{\hat{a}}, -c_{\hat{a}})$};
		
		\filldraw (4.99, -0.9) ellipse (0.08 and 0.04);
		\node[below] at (4.99, -0.95) {\footnotesize $(r_a, -c_a)$};
		
		\filldraw (2.86, 0) ellipse (0.08 and 0.04);
		\node[below] at (2.86, -0.05) {\footnotesize $r_{\hat{a}} - v^{-1}(c_{\hat{a}})$};
		
		\draw [dashed, color = blue, domain = (20/27):7.5, variable = \x]
		plot ({\x}, {40 / (27*\x)});
		
		\draw [domain = 0.672:7.5, variable = \x]
		plot ({\x}, {1.343 / \x});
		
		\filldraw[color = blue!50, fill = blue!50] (2.22, 0.667) ellipse (0.06 and 0.03);
		
		\filldraw[color = green!50!black, fill = green!50!black] (2.86, 0.463) ellipse (0.06 and 0.03);
		\node[right] at (3.06, 0.463) {\footnotesize $(x_1, v(r_a - x_1) - c_a)$};
		
		\node at (4.2, -2) {(b)};
	\end{tikzpicture}
	\caption{Optimal utility profile for a given action $a$ if the social planner uses Nash product as the social utility function and the agent is risk-averse. $x_a$ is the solution of the equation $v(r_a - x) - x \cdot v'(r_a - x) = c_a$. (a) If $(x_a, y_a)$ is implementable, it would be the optimal utility profile (the blue point). (b) If $(x_a, y_a)$ is not implementable, the optimal utility profile would be $(x_1, v(r_a - x_1) - c_a)$ (the green point).}
	\label{fig:optimization_np_risk_averse}
\end{figure}

\section{Optimal Utility Profile for a Given Action with Egalitarian Social Welfare}\label{sec:opt_profile_esf}
Figure~\ref{fig:optimization_esf_risk_neutral} and \ref{fig:optimization_esf_risk_averse} illustrate the optimal utility profile if the social planner uses egalitarian social welfare as the social utility function and the agent is risk-neutral and risk-averse, respectively.

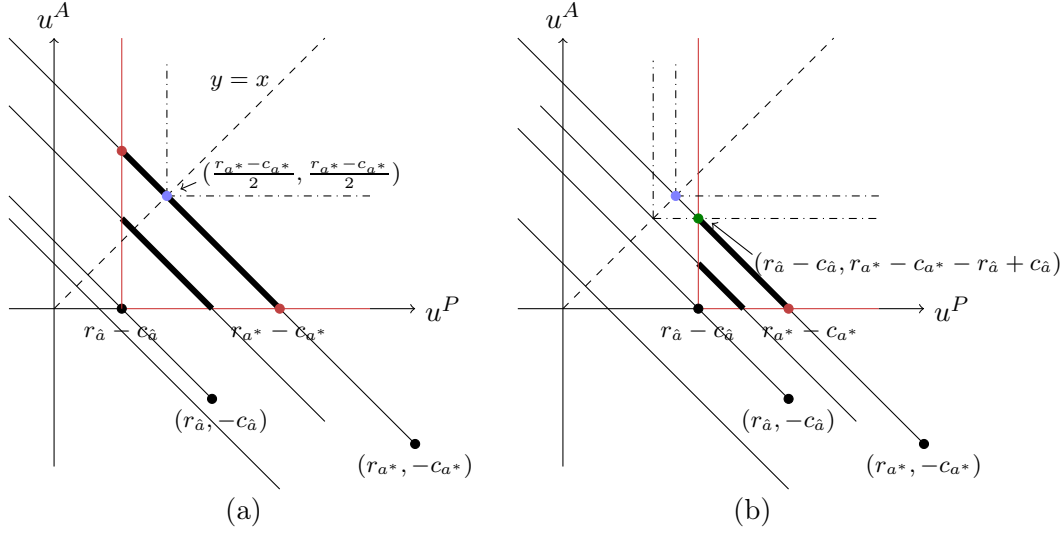
\begin{figure}[!htp]
	\centering
	\begin{tikzpicture}[xscale=0.6, yscale=0.6]
		\draw (-1,0)--(1.5,0);
		\draw[red!60!gray] (1.5,0)--(7,0);
		\draw[->] (7,0)--(8,0);
		\node[right] at (8,0) {$u^P$};
		
		\draw[->] (0,-3.5)--(0,6);
		\node[above] at (0,6) {$u^A$};
		
		\draw (3.5,-2)--(-1,2.5);
		\draw (5,-4)--(-1,2);
		\draw (6,-2.5)--(-1,4.5);
		\draw (8,-3)--(-1,6);
		
		\draw[line width=0.8mm] (3.5,0)--(1.5,2);
		
		\draw[line width=0.8mm] (5,0)--(1.5,3.5);
		
		\filldraw (3.5,-2) circle(0.1);
		\node[below] at (3.7,-2) {\footnotesize $(r_{\hat{a}}, -c_{\hat{a}})$}; 
		
		\filldraw (8,-3) circle(0.1);
		\node[below] at (8,-3) {\footnotesize $(r_{a^*}, -c_{a^*})$};
		
		\filldraw (1.5,0) circle(0.1);
		\node[below] at (1.5,-0.1) {\footnotesize $r_{\hat{a}} - c_{\hat{a}}$};
		
		\filldraw[color = red!50!gray, fill = red!50!gray] (5, 0) circle(0.1);
		\node[below] at (5,-0.1) {\footnotesize $r_{a^*} - c_{a^*}$};
		
		\filldraw[color = red!50!gray, fill = red!50!gray] (1.5,3.5) circle(0.1);
		\draw [dashdotted] (2.5, 2.5)--(7,2.5);
		\draw [dashdotted] (2.5, 2.5)--(2.5,5.5);

		\draw[red!60!gray] (1.5,0)--(1.5,6);
		\draw[dashed] (0,0)--(6,6);
		\node [left] at (5,5) {\footnotesize $y = x$};
		
		\filldraw[color = blue!50, fill = blue!50] (2.5, 2.5) circle(0.1);
		\node[right] at (3.05,3) {\footnotesize $(\frac{r_{a^*} - c_{a^*}}{2}, \frac{r_{a^*} - c_{a^*}}{2})$};
		\draw[->] (3.15,2.8)--(2.8,2.65);

		\node at (4.2, -4.5) {(a)};
	\end{tikzpicture}
	\hspace{2mm}
	\begin{tikzpicture}[xscale=0.6, yscale=0.6]
		\draw (-1,0)--(3,0);
		\draw[red!60!gray] (3,0)--(7,0);
		\draw[->] (7,0)--(8,0);
		\node[right] at (8,0) {$u^P$};
		
		\draw[->] (0,-3.5)--(0,6);
		\node[above] at (0,6) {$u^A$};
		
		\draw (5,-2)--(-1,4);
		\draw (5,-4)--(-1,2);
		\draw (6.5,-2.5)--(-0.5,4.5);
		\draw (8,-3)--(-1,6);
		
		\draw[line width=0.8mm] (4,0)--(3,1);
		
		\draw[line width=0.8mm] (5,0)--(3,2);
		
		\draw [dashed] (0,0)--(6,6);
		
		\filldraw (5,-2) circle(0.1);
		\node[below] at (5,-2) {\footnotesize $(r_{\hat{a}}, -c_{\hat{a}})$}; 
		
		\filldraw (8,-3) circle(0.1);
		\node[below] at (8,-3) {\footnotesize $(r_{a^*}, -c_{a^*})$};
		
		\filldraw[color = red!50!gray, fill = red!50!gray] (5, 0) circle(0.1);
		\node[below] at (5.5,-0.1) {\footnotesize $r_{a^*} - c_{a^*}$};
		
		\draw[red!60!gray] (3,0)--(3,6);
		
		\filldraw (3,0) circle(0.1);
		\node[below] at (3,-0.1) {\footnotesize $r_{\hat{a}} - c_{\hat{a}}$};
		
		\draw [dashdotted] (2.5,2.5)--(7,2.5);
		\draw [dashdotted] (2.5,2.5)--(2.5,5.5);
		
		\draw [dashdotted] (2,2)--(7,2);
		\draw [dashdotted] (2,2)--(2,5.5);
		
		\filldraw[color = green!50!black, fill = green!50!black] (3, 2) circle(0.1);
		\filldraw [color = blue!50, fill = blue!50] (2.5, 2.5) circle(0.1);
		
		\node [right] at (4,1) {\footnotesize $(r_{\hat{a}} - c_{\hat{a}}, r_{a^*} - c_{a^*} - r_{\hat{a}} + c_{\hat{a}})$};
		\draw [->] (4.2, 1.2)--(3.3,1.9);
		
		\node at (4.2, -4.5) {(b)};
	\end{tikzpicture}
	\caption{Optimal utility profile if the social planner uses egalitarian social welfare as the social utility function and the agent is risk-neutral. Denote $a^*$ as the action that induces the largest utilitarian social welfare (USF). The dashdotted line is the contour line for the objective function $\min \left\{u^P, u^A\right\}$. (a) $\frac{r_{a^*} - c_{a^*}}{2} \geq \max \left\{0, r_{\hat{a}} - c_{\hat{a}}\right\}$, the blue point represents the optimal utility profile. (b) $0 \leq \frac{r_{a^*} - c_{a^*}}{2} \leq r_{\hat{a}} - c_{\hat{a}}$, the green point represents the optimal utility profile.}
	\label{fig:optimization_esf_risk_neutral}
\end{figure}

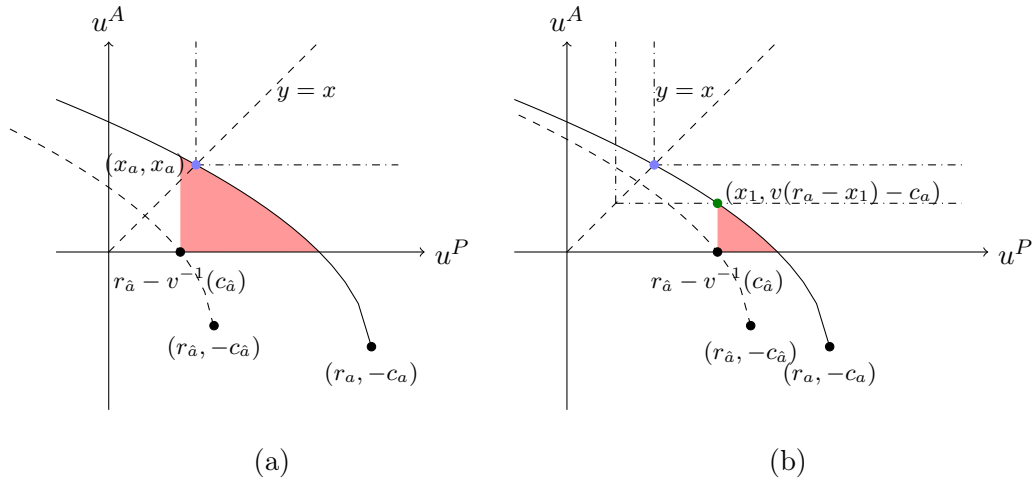
\begin{figure}[!htp]
	\centering
	\begin{tikzpicture}[xscale=0.7, yscale=1.4]
		\fill [red!40, domain = 1.001:3.64, variable = \x]
		(1.36, 0)--
		plot ({-\x+5}, {sqrt{\x} - 1})--cycle;
		
		\draw[->] (-1, 0)--(6,0);
		\draw[->] (0, -1.5)--(0,2);
		\node[right] at (6,0) {$u^P$};
		\node[above] at (0,2) {$u^A$};
		
		\draw[dashed, domain = 0.01:4, variable = \x]
		plot ({-\x+2}, {sqrt{\x} - 0.8});
		
		\draw[domain = 0.01:6, variable = \x]
		plot ({-\x+5}, {sqrt{\x} - 1});
		
		\filldraw (2, -0.7) ellipse (0.08 and 0.04);
		\node[below] at (2, -0.7) {\footnotesize $(r_{\hat{a}}, -c_{\hat{a}})$};
		
		\filldraw (4.99, -0.9) ellipse (0.08 and 0.04);
		\node[below] at (4.99, -0.95) {\footnotesize $(r_a, -c_a)$};
		
		\filldraw (1.36, 0) ellipse (0.08 and 0.04);
		\node[below] at (1.36, -0.05) {\footnotesize $r_{\hat{a}} - v^{-1}(c_{\hat{a}})$};
		
		\draw [dashed] (0, 0)--(4,2);
		
		\draw [dashdotted] (1.657,0.828)--(5.5,0.828);
		\draw [dashdotted] (1.657,0.828)--(1.657,2); 
		
		\filldraw[color = blue!50, fill = blue!50] (1.657,0.828) ellipse (0.08 and 0.04);
		\node[left] at (1.657,0.828) {\footnotesize $(x_a, x_a)$};
		
		\node[right] at (3,1.5) {\footnotesize $y = x$};
		
		\node at (3.1, -2) {(a)};
	\end{tikzpicture}
	\hspace{2mm}
	\begin{tikzpicture}[xscale=0.7, yscale=1.4]
		\fill [red!40, domain = 1.001:2.14, variable = \x]
		(2.86, 0)--
		plot ({-\x+5}, {sqrt{\x} - 1})--cycle;
		
		\draw[->] (-1, 0)--(8,0);
		\draw[->] (0, -1.5)--(0,2);
		\node[right] at (8,0) {$u^P$};
		\node[above] at (0,2) {$u^A$};
		
		\draw[dashed, domain = 0.01:4.5, variable = \x]
		plot ({-\x+3.5}, {sqrt{\x} - 0.8});
		
		\draw[domain = 0.01:6, variable = \x]
		plot ({-\x+5}, {sqrt{\x} - 1});
		
		\filldraw (3.49, -0.7) ellipse (0.08 and 0.04);
		\node[below] at (3.49, -0.75) {\footnotesize $(r_{\hat{a}}, -c_{\hat{a}})$};
		
		\filldraw (4.99, -0.9) ellipse (0.08 and 0.04);
		\node[below] at (4.99, -0.95) {\footnotesize $(r_a, -c_a)$};
		
		\filldraw (2.86, 0) ellipse (0.08 and 0.04);
		\node[below] at (2.86, -0.05) {\footnotesize $r_{\hat{a}} - v^{-1}(c_{\hat{a}})$};
		
		\draw [dashed] (0,0)--(4,2);
		\node [left] at (3,1.5) {\footnotesize $y = x$};
		
		\draw [dashdotted] (1.657,0.828)--(7.5,0.828);
		\draw [dashdotted] (1.657,0.828)--(1.657,2); 
		
		\draw [dashdotted] (0.926,0.463)--(7.5,0.463);
		\draw [dashdotted] (0.926,0.463)--(0.926,2);
		
		\filldraw[color = green!50!black, fill = green!50!black] (2.86, 0.463) ellipse (0.08 and 0.04);
		\node[right] at (2.8, 0.563) {\footnotesize $(x_1, v(r_a - x_1) - c_a)$};
		
		\filldraw[color = blue!50, fill = blue!50] (1.657,0.828) ellipse (0.08 and 0.04);
		
		\node at (4.2, -2) {(b)};
	\end{tikzpicture}
	\caption{Optimal utility profile for a given action $a$ if the social planner uses egalitarian social welfare as the social utility function and the agent is risk-averse. Denote $x_1 = r_{\hat{a}} - v^{-1}(c_{\hat{a}})$, and $x_a$ being the solution of the equation $x = v(r_a - x) - c_a$. The dashdotted line is the contour line for the objective function $\min \left\{u^P, u^A\right\}$. (a) If $(x_a, x_a)$ is implementable, it would be the optimal utility profile (the blue point). (b) If $(x_a, y_a)$ is not implementable, the optimal utility profile would be $(x_1, v(r_a - x_1) - c_a)$ (the green point).}
	\label{fig:optimization_esf_risk_averse}
\end{figure}

\section{Optimal Utility Profile for a Given Action with Approximated Fairness}\label{sec:opt_profile_af}
Figure~\ref{fig:optimization_af_risk_neutral} and \ref{fig:optimization_af_risk_averse} illustrate the optimal utility profile if the social planner uses approximated fairness as the social utility function and the agent is risk-neutral and risk-averse, respectively.

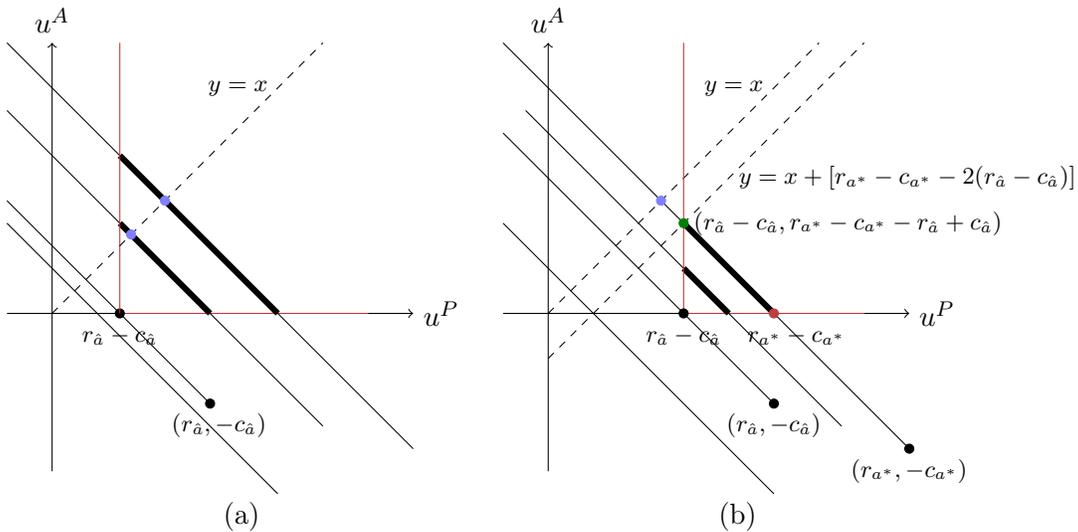
\begin{figure}[!htp]
	\centering
	\begin{tikzpicture}[xscale=0.6, yscale=0.6]
		\draw (-1,0)--(1.5,0);
		\draw[red!60!gray] (1.5,0)--(7,0);
		\draw[->] (7,0)--(8,0);
		\node[right] at (8,0) {$u^P$};
		
		\draw[->] (0,-3.5)--(0,6);
		\node[above] at (0,6) {$u^A$};
		
		\draw (3.5,-2)--(-1,2.5);
		\draw (5,-4)--(-1,2);
		\draw (6,-2.5)--(-1,4.5);
		\draw (8,-3)--(-1,6);
		
		\draw[line width=0.8mm] (3.5,0)--(1.5,2);
		
		\draw[line width=0.8mm] (5,0)--(1.5,3.5);
		
		\filldraw (3.5,-2) circle(0.1);
		\node[below] at (3.7,-2) {\footnotesize $(r_{\hat{a}}, -c_{\hat{a}})$}; 
				
		\filldraw (1.5,0) circle(0.1);
		\node[below] at (1.5,-0.1) {\footnotesize $r_{\hat{a}} - c_{\hat{a}}$};

		\draw[red!60!gray] (1.5,0)--(1.5,6);
		\draw[dashed] (0,0)--(6,6);
		\node [left] at (5,5) {\footnotesize $y = x$};
		
		\filldraw[color = blue!50, fill = blue!50] (2.5, 2.5) circle(0.1);
		\filldraw [color = blue!50, fill = blue!50] (1.75, 1.75) circle(0.1);

		\node at (4.2, -4.5) {(a)};
	\end{tikzpicture}
	\hspace{2mm}
	\begin{tikzpicture}[xscale=0.6, yscale=0.6]
		\draw (-1,0)--(3,0);
		\draw[red!60!gray] (3,0)--(7,0);
		\draw[->] (7,0)--(8,0);
		\node[right] at (8,0) {$u^P$};
		
		\draw[->] (0,-3.5)--(0,6);
		\node[above] at (0,6) {$u^A$};
		
		\draw (5,-2)--(-1,4);
		\draw (5,-4)--(-1,2);
		\draw (6.5,-2.5)--(-0.5,4.5);
		\draw (8,-3)--(-1,6);
		
		\draw[line width=0.8mm] (4,0)--(3,1);
		
		\draw[line width=0.8mm] (5,0)--(3,2);
		
		\draw [dashed] (0,0)--(6,6);
		\node [left] at (5,5) {\footnotesize $y = x$};
		
		\draw [dashed] (0,-1)--(7,6);
		\node [right] at (4,3) {\footnotesize $y = x + [r_{a^*} - c_{a^*} - 2 (r_{\hat{a}} - c_{\hat{a}})]$};
		
		\filldraw (5,-2) circle(0.1);
		\node[below] at (5,-2) {\footnotesize $(r_{\hat{a}}, -c_{\hat{a}})$}; 
		
		\filldraw (8,-3) circle(0.1);
		\node[below] at (8,-3) {\footnotesize $(r_{a^*}, -c_{a^*})$};
		
		\filldraw[color = red!50!gray, fill = red!50!gray] (5, 0) circle(0.1);
		\node[below] at (5.5,-0.1) {\footnotesize $r_{a^*} - c_{a^*}$};
		
		\draw[red!60!gray] (3,0)--(3,6);
		
		\filldraw (3,0) circle(0.1);
		\node[below] at (3,-0.1) {\footnotesize $r_{\hat{a}} - c_{\hat{a}}$};

		\filldraw[color = green!50!black, fill = green!50!black] (3, 2) circle(0.1);
		\node[right] at (3,2) {\footnotesize $(r_{\hat{a}} - c_{\hat{a}}, r_{a^*} - c_{a^*} - r_{\hat{a}} + c_{\hat{a}})$};
		
		\filldraw [color = blue!50, fill = blue!50] (2.5, 2.5) circle(0.1);
		
		\node at (4.2, -4.5) {(b)};
	\end{tikzpicture}
	\caption{Optimal utility profile if the social planner uses approximated fairness as the social utility function and the agent is risk-neutral. Denote $a^*$ as the action that induces the largest utilitarian social welfare (USF). (a) If $y = x$ intersects the implementable set with more than one point, any intersection point corresponds to an optimal utility profiles (the blue points). (b) $y = x$ has no intersection with the implementable set, the green point represents the optimal utility profile.}
	\label{fig:optimization_af_risk_neutral}
\end{figure}

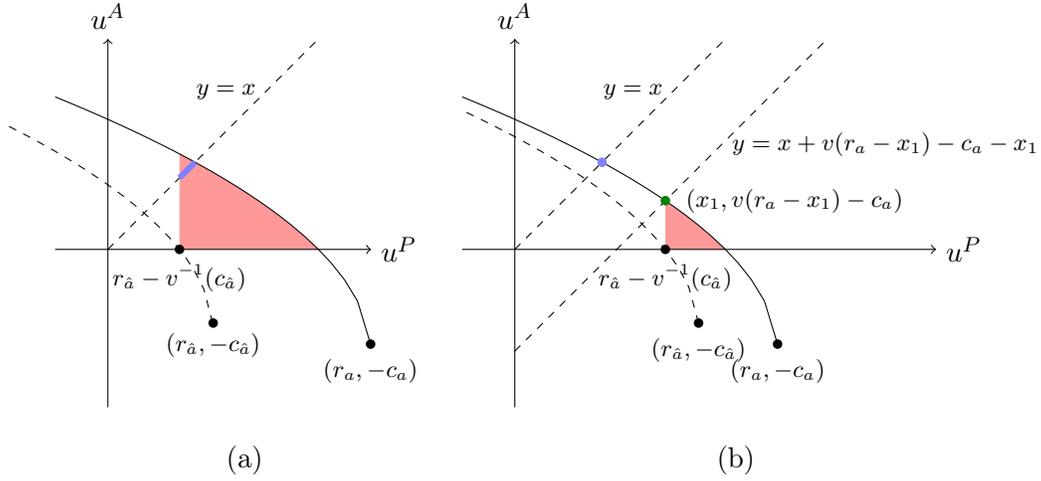
\begin{figure}[!htp]
	\centering
	\begin{tikzpicture}[xscale=0.7, yscale=1.4]
		\fill [red!40, domain = 1.001:3.64, variable = \x]
		(1.36, 0)--
		plot ({-\x+5}, {sqrt{\x} - 1})--cycle;
		
		\draw[->] (-1, 0)--(5,0);
		\draw[->] (0, -1.5)--(0,2);
		\node[right] at (5,0) {$u^P$};
		\node[above] at (0,2) {$u^A$};
		
		\draw[dashed, domain = 0.01:4, variable = \x]
		plot ({-\x+2}, {sqrt{\x} - 0.8});
		
		\draw[domain = 0.01:6, variable = \x]
		plot ({-\x+5}, {sqrt{\x} - 1});
		
		\filldraw (2, -0.7) ellipse (0.08 and 0.04);
		\node[below] at (2, -0.7) {\footnotesize $(r_{\hat{a}}, -c_{\hat{a}})$};
		
		\filldraw (4.99, -0.9) ellipse (0.08 and 0.04);
		\node[below] at (4.99, -0.95) {\footnotesize $(r_a, -c_a)$};
		
		\filldraw (1.36, 0) ellipse (0.08 and 0.04);
		\node[below] at (1.36, -0.05) {\footnotesize $r_{\hat{a}} - v^{-1}(c_{\hat{a}})$};
		
		\draw [dashed] (0, 0)--(4,2);
		
		\draw [color = blue!50, line width=0.8mm] (1.657,0.828)--(1.36,0.68);
		
		\node[left] at (3,1.5) {\footnotesize $y = x$};
		
		\node at (2.6, -2) {(a)};
	\end{tikzpicture}
	\hspace{2mm}
	\begin{tikzpicture}[xscale=0.7, yscale=1.4]
		\fill [red!40, domain = 1.001:2.14, variable = \x]
		(2.86, 0)--
		plot ({-\x+5}, {sqrt{\x} - 1})--cycle;
		
		\draw[->] (-1, 0)--(8,0);
		\draw[->] (0, -1.5)--(0,2);
		\node[right] at (8,0) {$u^P$};
		\node[above] at (0,2) {$u^A$};
		
		\draw[dashed, domain = 0.01:4.5, variable = \x]
		plot ({-\x+3.5}, {sqrt{\x} - 0.8});
		
		\draw[domain = 0.01:6, variable = \x]
		plot ({-\x+5}, {sqrt{\x} - 1});
		
		\filldraw (3.49, -0.7) ellipse (0.08 and 0.04);
		\node[below] at (3.49, -0.75) {\footnotesize $(r_{\hat{a}}, -c_{\hat{a}})$};
		
		\filldraw (4.99, -0.9) ellipse (0.08 and 0.04);
		\node[below] at (4.99, -0.95) {\footnotesize $(r_a, -c_a)$};
		
		\filldraw (2.86, 0) ellipse (0.08 and 0.04);
		\node[below] at (2.86, -0.05) {\footnotesize $r_{\hat{a}} - v^{-1}(c_{\hat{a}})$};
		
		\draw [dashed] (0,0)--(4,2);
		\node [left] at (3,1.5) {\footnotesize $y = x$};
		
		\draw [dashed] (0,-0.967)--(5.934,2);
		\node [right] at (3.934,1) {\footnotesize $y = x + v(r_a - x_1) - c_a - x_1$};
		
		\filldraw[color = green!50!black, fill = green!50!black] (2.86, 0.463) ellipse (0.08 and 0.04);
		\node[right] at (3.06, 0.463) {\footnotesize $(x_1, v(r_a - x_1) - c_a)$};
		
		\filldraw[color = blue!50, fill = blue!50] (1.657,0.828) ellipse (0.08 and 0.04);
		
		\node at (4.2, -2) {(b)};
	\end{tikzpicture}
	\caption{Optimal utility profile for a given action $a$ if the social planner uses approximated fairness as the social utility function and the agent is risk-averse. Denote $x_1 = r_{\hat{a}} - v^{-1}(c_{\hat{a}})$. (a) If $y = x$ intersects the implementable set, the bold blue parts of the line represents optimal utility profiles. (b) If $y = x$ does not intersect the implementable set, the green point represents the optimal utility profile.}
	\label{fig:optimization_af_risk_averse}
\end{figure}

\section{Utility Profiles for Risk-Averse Principal and Risk-Neutral Agent}\label{sec:profile_risk_averse_principal}
Figure~\ref{fig:extension_risk_averse} shows the possible utility profiles when the principal is risk-averse and the agent is risk-neutral.

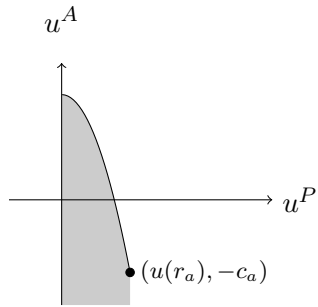
\begin{figure}[!htp]
	\centering
	\begin{tikzpicture}[xscale=0.7, yscale=1.4]
		\fill [gray!40, domain = 0.01:1.3, variable = \x]
		(1.3, -1)--
		(0, -1)--
		plot ({\x}, {-\x^2 + 1})--cycle;
		
		\draw[->] (-1, 0)--(4,0);
		\draw[->] (0, -1)--(0,1.3);
		\node[right] at (4,0) {$u^P$};
		\node[above] at (0,1.5) {$u^A$};

		\draw[domain = 0.01:1.3, variable = \x]
		plot ({\x}, {-\x^2 + 1});
		
		\filldraw (1.3, -0.69) ellipse (0.08 and 0.04);
		\node[right] at (1.3, -0.69) {\footnotesize $(u(r_a), -c_a)$};
		
	\end{tikzpicture}
	\caption{Super set of possible utility profiles when the principal is risk-averse while the agent is risk-neutral.}
	\label{fig:extension_risk_averse}
\end{figure}

\section{Discussion and Extension}\label{sec:extend}
We discuss possible extensions of our framework and the relationship between our model and conventional principal-agent model, and Bayesian persuasion.

\subsection{Generalization to Online Learning}\label{subsec:online_learning}
The principal-agent model considered in this paper is completely transparent, where the cost for each action, the reward for each outcome, and the value function for the monetary transfer are common knowledge of the game. However, some of the assumptions are stringent that may limit its applicability in practice. For example, the risk attitude and the specific value function of the agent may be private.

Concerning the above issue, we can assume that the players can interact repeatedly and cast the problem into an online learning framework. In this setting, the agent may have different risk attitudes, while the value function towards the contract is unknown and is chosen adversarially from a finite set of possible types, the social planner prescribes an information structure at each round and we are interested in developing \emph{no-regret} algorithms with performances comparable to a best-in-hindsight information structure. Such flavour of work has been considered in the model of Bayesian persuasion~\citep{castiglioni2020online,castiglioni2021multi,bernasconi2023optimal},but is lack of study in the information design in principal-agent problem, and we leave it as future work.

\subsection{Selfish Third Party}\label{subsec:selfish_third_party}
In this work, we assume that the social planner is a completely independent third party, whose utility function is solely based on the social purpose. In reality, there may be situations where the third party has its own concern or personal interest. In this case, the optimization problem in the first stage should be modified for the consideration of the third party's utility, and the geometric interpretation of the utility functions should be triaxial as well, which may lead to different optimal utility profiles. We leave such consideration as future work.

\subsection{Relationship with Principal-Agent Problem}\label{subsec:relation_principal_agent}
In economic theory, the principal-agent problem typically arises where the two parties have different interests and asymmetric information. Concerning the information asymmetry, the model can be divided into two categories: (a) moral hazard~\citep{holmstrom1979moral} where the actions of the agent is not observed and (b) adverse selection~\citep{hart1987theory} where the characteristics of the agent is not observed. 

The setting considered in this paper is closely related to the moral hazard model, where the agent takes some actions which the principal cannot observe, but instead some signals from those actions are revealed and contracts should be written on those signals.
However, there are two main differences. 
First, in classic principal-agent model, the signal not only provides the principal with some information about the agent's action, but also specifies the profit or reward for the principal, while the information structure in our model has no role in determining the income for the principal. 
Second, classic principal-agent model involves two parties and mainly focus on designing the optimal contract from the principal's point of view, assuming that the agent is rational.
In this work, we introduce a third party, the social planner, who acts as a conciliator between the principal and the agent. By controlling the information flow, the social planner has great power to decide ``where the game is going''. The main purpose of our work is to design the information structure toward a specific social purpose characterized by a social utility function, assuming that both the principal and the agent are rational that each of them would maximize their own utility function. The workflow proposed in this work provides insight for the minimum amount of information needed to induce certain equilibrium of the system. 
Whether the information design process developed in this work can be extended to the model of adverse selection problem, i.e., the system is not fully known to the social planner while the principal and/or the agent may be informed some private messages before the information structure is designed, would be an interesting future direction.

\subsection{Relationship with Bayesian Persuasion}\label{subsec:relation_bp}
There are several fundamental differences between information design in the principal-agent problem and the seminal work of Bayesian persuasion~\citep{kamenica2011bayesian}.

In Bayesian persuasion, the model consists of two players, a sender and a receiver, where the sender is the information structure (signaling scheme) designer. Owing to the superiority of knowing the realized state of nature, the sender designs the signaling scheme such that the receiver would take an action maximizing the sender's payoff, where the model assumes that the sender and the receiver shares a common prior distribution over the states of nature. In the information design in principal-agent problem, there are three players, the social planner, the principal and the agent, where the social planner designs the information structure, the principal then designs a contract based on the information structure and finally the agent takes the action which maximizes her own utility function and affects the payoff of both the principal and the society. In Bayesian persuasion, the signaling scheme is designed for each possible state of nature, while in principal-agent problem, the information structure is depicted for each available actions for the agent. A fundamental difference between these two scenarios is that the Bayesian persuasion sender is merely an observer and has no power to change the state of nature, while in the principal-agent model, different information structures induce different contract designed by the principal, and would in turns induce different optimal actions for the agent. Therefore, in some sense, the information design problem is more complicated than Bayesian persuasion.

However, in another sense, the information structure is simpler in the principal-agent model than in the Bayesian persuasion model. As is noted in \citet[Proposition 1]{kamenica2011bayesian}, the signaling scheme in Bayesian persuasion has a close relationship to the revelation principle~\citep{myerson1979incentive}, i.e., the sender only needs to design signaling scheme that directly recommends actions to the receiver. Therefore, in the basic model of Bayesian persuasion, researchers usually assume that the signal space is not smaller than the state space and the action space~\citep{kamenica2019bayesian}, i.e., $\abs{S} \geq \min \left\{\abs{\Omega}, \abs{A}\right\}$. This would raise an issue if the state space and the action space are both large while the availability of messages is limited, \citet{aybas2019persuasion} shows that the sender's utility would always be worse off with coarse communication. Fortunately, in the information design in principal-agent problem, such difficulty need hardly be taken into account. Although there are infinite available contracts for the principal and $n$ possible actions for the agent, the social planner needs only care about binary-signal information structure. Such ``blessing'' comes from the specific structure of the principal-agent problem. The agent's utility function consists of two parts, i.e., the cost for the chosen action and the expected monetary transfer. Therefore, if two actions share the same expected transfer, the one with cheaper cost always dominates, and hence we can simply divide the actions into two groups: one group employs the socially optimal action $a^*$, i.e., the corresponding action for the optimal utility profile that the social planner would like to induce, as the dominated action, while the other group employs action $\hat{a}$, which by definition is the action with the maximum expected reward for the principal among all the least costly actions for the agent, as the dominant. From this observation, it suffices for the principal to design a contract with binary-choice monetary transfer, since the size of the ``state space'' is two.

\end{document}